\documentclass[12pt]{article}
\usepackage{amsthm}
\usepackage{amsmath}
\usepackage{ bbold }
\usepackage[colorlinks,citecolor=blue,urlcolor=blue,filecolor=blue,backref=page]{hyperref}
\usepackage{graphicx}
\usepackage[english]{babel}
\usepackage[utf8]{inputenc}
\usepackage{amsmath}
\usepackage{amssymb}
\usepackage{graphicx}
\usepackage{hyperref}
\usepackage{graphicx}
\usepackage{comment}
\usepackage{mathrsfs}
\usepackage{amsthm}
\usepackage{bbm}
\usepackage{scalerel}
\usepackage{rotating}
\usepackage{lscape}
\usepackage{subfigure}
\usepackage[colorinlistoftodos]{todonotes}
\usepackage{lipsum}
\usepackage{comment}
\usepackage[round]{natbib}
\usepackage[colorlinks,citecolor=blue,urlcolor=blue,filecolor=blue,backref=page]{hyperref}
\usepackage{graphicx}
\newtheorem{theorem}{Theorem}

\newtheorem{lemma}[theorem]{Lemma}

\numberwithin{equation}{section}
\theoremstyle{plain}

\newcommand{\ignore}[1]{}

\newcommand{\blind}{0}

\makeatother
\newcommand{\E}{\mathrm{E}}
\newcommand{\N}{\mathrm{N}}
\renewcommand{\P}{\mathrm{P}}
\newcommand{\RR}{\mathbb{R}}

\newcommand{\HH}{\mathbb{H}}
\newcommand{\KK}{\mathbb{K}}
\newcommand{\LL}{\mathbb{L}}
\newcommand{\trans}{^{\tiny{\mathrm{T}}}}

\begin{document}

\def\spacingset#1{\renewcommand{\baselinestretch}%
{#1}\small\normalsize} \spacingset{1}

%
%
%
%
%
\if0\blind
{
  \title{\bf {Optimal Bayesian Smoothing of Functional Observations over a Large Graph}}
  \author{$^1$Arkaprava Roy, 
       $^2$Shubhashis Ghosal\\
       $^1$University of Florida, $^2$North Carolina State University}
  \maketitle
} \fi

\if1\blind
{
  \bigskip
  \bigskip
  \bigskip
  \begin{center}
    {\LARGE\bf }
\end{center}
  \medskip
} \fi

\bigskip
	
	\begin{abstract}
		In modern contexts, some types of data are observed in high-resolution, essentially continuously in time. Such data units are best described as taking values in a space of functions. Subject units carrying the observations may have intrinsic relations among themselves, and are best described by the nodes of a large graph. It is often sensible to think that the underlying signals in these functional observations vary smoothly over the graph, in that neighboring nodes have similar underlying signals. This qualitative information allows borrowing of strength over neighboring nodes and consequently leads to more accurate inference. In this paper, we consider a model with Gaussian functional observations and adopt a Bayesian approach to smoothing over the nodes of the graph. We characterize the minimax rate of estimation in terms of the regularity of the signals and their variation across nodes quantified in terms of the graph Laplacian. We show that an appropriate prior constructed from the graph Laplacian can attain the minimax bound, while using a mixture prior, the minimax rate up to a logarithmic factor can be attained simultaneously for all possible values of functional and graphical smoothness.  We also show that in the fixed smoothness setting, an optimal sized credible region has arbitrarily high frequentist coverage. A simulation experiment demonstrates that the method performs better than potential competing methods like the random forest. The method is also applied to a dataset on daily temperatures measured at several weather stations in the US state of North Carolina. 
		\end{abstract}
		
		\noindent%
	\spacingset{1.45}
	
\section{Introduction}
\label{sec:introduction}

Functional observations at various locations are commonly encountered in spatial statistics, and are often called spatio-temporal data. Extracting the underlying signal from such noisy spatio-temporal data involves smoothing over both space and time. In the Bayesian context, Gaussian processes are often used to describe spatio-temporal measurements (\citet{banerjee2014hierarchical}). Functional data can also be associated with subjects, for instance, in medical or financial applications. In the modern era, the internet traffic data, or data obtained from wearable devices, are functional in nature, and are typically collected from many sources simultaneously. Unlike in the spatial context, a clear notion of the closeness of physical locations over which a smoothing can be carried out may be lacking, but some other notion of connectivity may be meaningful. Such a complex relational network may be described by a graph, with subjects standing for the nodes of the graph, and an edge connecting two nodes stands for a relation. Typically, the underlying graph has a large number of nodes. In most contexts, the graph is given or is easily identified. For instance, in spatial statistics, closeness of physical locations may clearly define neighboring nodes and constitute a graph, while in other contexts such as a protein interaction network (\citet{sharan2007network}) or an image interpolation problem (\citet{liu2013image}), the graph structure comes from the nature of the problem. In some other contexts such as voting patterns by politicians on issues (\citet{poole1991patterns}), the graph may have to be estimated from the data using models for interaction such as the Ising model (\citet{cipra1987introduction}). 

When the observations at the nodes of the graph are scalar measurements, an asymptotic framework for assessing the estimation of their parameters was proposed by \citet{kirichenko2018minimax} in terms of the so-called graph-Laplacian. The growth of the eigenvalues of the graph-Laplacian leads to a notion of a dimension of the graph. \citet{kirichenko2018minimax} introduced smoothness classes determined by a weighted Euclidean distance on the parameter vector using an appropriate power of the graph-Laplacian, depending on the dimension of the graph. Then they obtained the minimax rates of estimation in these smoothness classes under infill asymptotics. \citet{kirichenko2017estimating} developed a Bayesian procedure using a multivariate normal prior with an appropriate dispersion matrix, and showed that the resulting Bayesian procedure achieves the minimax rate of estimation. Further, as the prior does not use the knowledge of smoothness, the obtained rate automatically adapts to the smoothness. 

In this paper, we consider functional observations over a domain $\mathbb{T}$ at the nodes of a graph. The goal is to optimally recover their underlying mean functions. The functional observations are assumed to be distributed according to Gaussian processes on $\LL_2(\mathbb{T})$, the space of square-integrable functions on $\mathbb{T}$. Multidimensional observations at nodes are treated by considering a finite set $\mathbb{T}$, while a continuous domain gives functional observations. We characterize the minimax rate for inference on the vector of means taking values in a certain type of smoothness classes in $(\LL_2(\mathbb{T}))^n$. We then construct a Bayesian procedure using a joint Gaussian process prior with covariance kernel determined by an appropriate power of the graph-Laplacian when the smoothness level is given. We show that the posterior contraction rate equals the minimax rate, thus showing the asymptotic optimality of the proposed Bayesian procedure in the frequentist sense.  We note that as the function becomes infinitely smooth, we recover the rate obtained by \citet{kirichenko2017estimating} for scalar observations.  As the graphical smoothness increases indefinitely, the classical setting with independent and identically distributed (i.i.d.) replications of functional observations is approached and standard rates for one-dimensional function smoothing are recovered.  By using a random series prior with an unspecified number of terms, we show that the Bayesian procedure can adapt to the unknown smoothness, in that a single prior achieves the obtained rate simultaneously for all levels of smoothness, within a logarithmic factor. Then, we show that in a fixed smoothness setting, with an optimal choice of the prior, a posterior credible ball inflated by an appropriate constant, has frequentist coverage converging to one. Therefore, the resulting credible ball, which is easy to obtain by posterior simulation, can serve as a confidence ball in the frequentist sense, thus justifying Bayesian uncertainty quantification in the frequentist sense. The result is new even in the context of scalar observations. Finally, a posterior contraction result for discrete domain quantifying the accuracy of estimating trend of a multivariate time series is presented.  

The following notations will be used throughout the paper. The symbols `$\lesssim$' `$\gtrsim$' will stand for inequality up to an unimportant constant multiple, and $\asymp$ for the equality of the order of magnitude. For a vector $a$ (respectively matrix $\mathrm{A}$), $a\trans$ (respectively, $\mathrm{A}\trans$) will denote the transpose. Let $\mathrm{I}$ stand for the $n\times n$ identity matrix. The indicator function is denoted by $\mathbbm{1}$ and the number of elements of a finite set by $\#$. 

The paper is organized as follows. In the next section, we describe the model, present the preliminaries on graph-Laplacian and multidimensional Gaussian processes, and introduce the prior. In Section~\ref{sec:main}, the minimax rate is characterized, posterior contraction rates are obtained in both known and unknown smoothness settings, and asymptotic coverage of an appropriate credible ball is obtained.  A simulation study comparing the performance of the proposed procedure with the random forest and a parallelly implemented univariate time series imputation algorithm is presented in Section~\ref{sec:simulation}. The proposed method is also illustrated with a daily temperature data at weather stations in North Carolina. Proofs are given in the appendix.

\section{Model, prior and preliminaries}
\label{sec:model}

Let $G=(V,E)$ stand for a connected undirected graph with the set of nodes $V$, where $E\subset V\times V$ is the set of edges not containing any diagonal element. Suppose that with each node $i\in V$, there is an associated characteristic $f_i$, which is assumed to ``change gradually over neighboring nodes''. The notion of a gradual change can be made precise when the cardinality $n$ of $V$ is large, in a manner similar to in-fill asymptotics used on a lattice. However, unlike the latter case where grid-points are regularly placed and the geometry is homogeneous, a notion of smoothness on a graph should take the structure of the graph in consideration. This can be captured by the adjacency matrix $\mathrm{A}=(\!(a_{ij})\!)$ with $a_{ij}=\mathbbm{1}\{(i,j)\in E\}$. Let $\mathrm{D}=(\!(d_{ij})\!)$ stand for the diagonal matrix with $d_{ii}=\#\{j\in V: (i,j)\in E\}$, the degree of the $i$th node. Then the graph-Laplacian is defined by $\mathrm{ L}=\mathrm {D}-\mathrm {A}$. Clearly, $\mathrm{ L}$ is symmetric, and can be shown to be positive semi-definite with the minimum eigenvalue always $0$ (corresponding to the eigenvector $(1,\ldots,1)$) and all other eigenvalues are positive. Let $0=\lambda_0<\lambda_1\le \cdots \le \lambda_{n-1}$ stand for the eigenvalues of $\mathrm{ L}$. A common situation is that 
\begin{equation}
\label{Laplacian eigenvalues}
C_1 ((i-1)/n)^{2/r}\le \lambda_i \le C_2 ((i-1)/n)^{2/r}, \quad i=i_0,\ldots, \kappa n,
\end{equation}
for some constants $C_1,C_2$, positive integer $i_0$, $\kappa \in (0,1]$, and $r\ge 1$, known as the dimension of the graph. For instance, if the graph is a $d$-dimensional lattice, $r=d$, so $r$ generalizes the notion of dimension to a general graph. By a result of \citet{mohar1991laplacian}, $\lambda_1\ge 4/n^2$, so $r<1$ is not possible. The value of $r$ may be obtained numerically by regressing $\log \lambda_i$ against $\log i$, such as in the small world-graph (\citet{watts1998collective}) and a protein-interaction graph (\citet{kolaczyk2014statistical}), or for the dataset used in Section~\ref{sec:simulation}. As argued in \citet{kirichenko2017estimating}, we may assume that $\kappa=1$ in \eqref{Laplacian eigenvalues} at the expense of a larger value for the constant $C_2$.

When the characteristics associated with the nodes are real-valued, the object of interest is an $n$-dimensional column vector $f=(f_1,\ldots,f_n)\trans$, and the accuracy of estimation is measured by the normalized Euclidean norm $(n^{-1} \sum_{i=1}^n f_i^2)^{1/2}$, \citet{kirichenko2017estimating,kirichenko2018minimax} considered smoothness classes based on the graph-Laplacian as follows: a Sobolev ball of regularity $\beta$ and radius $Q$ is defined to be $H^\beta(Q)=\{f: f\trans (\mathrm{I}+(n^{2/r}\mathrm{ L})^\beta) f \le nQ^2 \}$. The primary motivation behind the choice is that if the graph is a one-dimensional lattice (so $r=1$) and $f_i=f(i/n)$ for some smooth function $f$ on $[0,1]$, then the regularity $\beta=1$ in the above sense corresponds to the Sobolev regularity of order $1$ for the function $f$. Under this setting, \citet{kirichenko2018minimax} showed that, based on independent observations $Y_i\sim \N (f_i, \sigma^2)$ where $\sigma$ is fixed or lies in a compact subinterval of $(0,\infty)$, the minimax rate of estimation in $H^\beta(Q)$ is $n^{-\beta/(2\beta+r)}$. This can be achieved by a multivariate normal prior on $f$: 
\begin{equation}
\label{graph-Laplacian prior}
f|c\sim \N_n(0, (c/n)^{(2\alpha+r)/r} (\mathrm{L}+ n^{-2} \mathrm{I})^{-(\alpha+r/2)}), \qquad c\sim \mathrm{Exp}(1),
\end{equation}
provided that $\beta\le \alpha+r/2$. The matrix powers are well-defined through the spectral theorem in view of the nonnegative definiteness of $\mathrm{L}$. Because the prior does not depend on the underlying regularity index $\beta$, the Bayesian method is automatically rate-adaptive, on the range $\beta \in (0, \alpha+r/2]$. They also showed that the full-range adaptation up to a logarithmic factor is possible using a different prior based on the exponential of the Laplacian. 

To generalize the results of \citet{kirichenko2017estimating,kirichenko2018minimax} to functional observations, we model $Y_1,\ldots,Y_n$ as independent Gaussian processes on $\LL_2(\mathbb{T})$ for a compact domain $\mathbb{T}$, with mean functions $f_1,\ldots,f_n$ respectively, and a common known covariance kernel $\sigma_n^2 \Sigma$. The kernel $\Sigma$ acts as a compact operator on $\LL_2(\mathbb{T})$, and hence has eigenvalues $\kappa_j\downarrow 0$. Let $\psi_j$ stand for the normalized eigenfunction of $\Sigma$ corresponding to the eigenvalue $\kappa_j$. Since Bayesian inference needs a likelihood function, the family of measures of $Y=(Y_1,\ldots,Y_n)\trans$ as $(f_1,\ldots,f_n)\trans$ varies over possible values must be absolutely continuous with respect to each other. This can hold only if $f_1,\ldots,f_n$ belong to the reproducing kernel Hilbert space (RKHS) $\HH$ of the covariance kernel $\Sigma$; see Appendix I of \citet{Ghosal}. The RKHS $\HH$ is  
a subspace of $\LL_2(\mathbb{T})$ consisting of functions $g=\sum_{j=1}^\infty \kappa_j^{1/2} \theta_j \psi_j$, where $\sum_{j=1}^\infty \theta_j^2<\infty$. The RKHS norm $\| \cdot \|_{\HH}$ on $\HH$ given by $\|g\|_{\HH}^2=\sum_{j=1}^\infty  \theta_j^2$ makes $\HH$ a Hilbert space. 

The space where $f=(f_1,\ldots,f_n)\trans$ takes values is $\HH^n$ equipped with the normalized aggregated norm $\|f\|_{n} =  (n^{-1} \sum_{i=1}^n \|f_i\|_{\HH}^2)^{1/2}$. We denote the corresponding inner product by $\langle \cdot, \cdot \rangle_{n}$. 
Write $f_i =\sum_{j=1}^\infty \kappa_j^{1/2} \theta_{ij} \psi_j$, $i=1,2,\ldots$. 
Also let  $\theta_j=(\theta_{1j}, \ldots,\theta_{nj})\trans\in \RR^n$ stand for the vector of the $j$th coefficients in the expansion of $f_1,\ldots,f_n$ in terms of the basis $\mathcal{B}=(\psi_1,\psi_2,\ldots)$. Note that $f$ has an eigen-representation $f=\sum_{i=0}^{n-1} \sum_{j=1}^\infty \kappa_j^{1/2} \vartheta_{ij} \psi_j e_i$, where $e_{0},\ldots,e_{n-1}$ are the eigenvectors of $\mathrm{L}$ corresponding to the eigenvalues $\lambda_0,\ldots,\lambda_{n-1}$. 
A equivalent canonical model is given by $Z_{ij}\sim \mathrm{N}( \vartheta_{ij}, 1)$ independently, where $Z_{ij}=\kappa_j^{-1/2} \int \langle Y, e_i \rangle \psi_j$ and $\vartheta_{ij}=\kappa_j^{-1/2} \int \langle f, e_i\rangle \psi_j$. Note that if $i\ne i'$, $Z_{ij}$ and $Z_{i'j'}$  are clearly independent for any $j,j'$, while $\mathrm{Cov} (Z_{ij},Z_{ij'})= \kappa_j^{-1/2} \kappa_{j'}^{-1/2} \int\!\int \Sigma(s,t) \psi_j(s) \psi_{j'}(t)ds\,dt= (\kappa_j/\kappa_{j'})^{1/2}\int \psi_j(t) \psi_{j'}(t)dt=   \mathbbm{1}(j=j')$, so $Z_{ij}$ and $Z_{ij'}$ are also  independent for $j\ne j'$.

The smoothness of a vector of functions on the graph can be described in terms of these coefficients. 
For $\beta,\gamma>0$, to quantify the regularity with respect to the graphical structure and the temporal direction respectively, define the $(\beta,\gamma)$-Sobolev ball $\mathcal{H}^{\beta,\gamma}(Q)$ of radius $Q$ in $\HH^n$ by 
$$
\big\{f: \sum_{j=1}^\infty  { j^{2\gamma}}  \langle \theta_j, (\mathrm{I}+(n^{2/r}\mathrm{L})^\beta) \theta_j\rangle_{n} \le Q^2 \big\}= \big\{f: n^{-1}\sum_{i=1}^{n}\sum_{j=1}^\infty { j^{2\gamma}}   (1+n^{2\beta/r}\lambda_i^\beta) \vartheta_{ij}^2 \le Q^2 \big\}. 
$$
To understand the notion, consider the white noise model $dY_i(t)=f_i(t)dt+dB_i(t)$ in the equivalent form $Y_i(t)=\int_0^t f_i(s)ds+B_i(t)$, where $B_i$ are independent standard Brownian motions. Then $\kappa_{2j-1}=\kappa_{2j}=j^{-1}$ and the smoothness $\gamma$ in the above sense coincides with the common Sobolev smoothness of each component function $f_i$. 

We note that an $n\times n$-matrix $\mathrm{ A}=(\!(a_{il})\!)$ can also be identified as a linear operator on $\HH^n$ through the relation $\mathrm{A} (h_1,\ldots,h_n)\trans= (\sum_{l=1}^n a_{1l} h_l,\ldots,\sum_{l=1}^n a_{nl} h_l) \trans$ for any $(h_1,\ldots,h_n)\trans\in \HH^n$. For an $n\times n$-matrix $\mathrm{ A}$ and a linear operator $B$ on $\HH$, define the Kronecker product $\mathrm{ A}\otimes B$ to be a linear operator on $\HH^n$ such that  $(\mathrm{ A}\otimes B)(h_1,\ldots,h_n)\trans= \mathrm{ A} (Bh_1,\ldots,Bh_n)\trans$. If $\mathrm{ A}$ is a (symmetric) nonnegative definite matrix and $B$ is a Hermitian nonnegative definite operator on $\HH$, then it is easy to verify that $\mathrm{ A}\otimes B$ is a Hermitian nonnegative definite operator on $\HH^n$. If $\psi,\phi\in \HH$, the tensor product $\psi\otimes \phi$ is a  linear operator on $\HH$ defined by $(\psi\otimes \phi)(h)=\langle \phi, h\rangle_{\HH} \psi$.   

We put a multivariate Gaussian process prior on $f\in \HH^n$ with a separable covariance operator given by a tensor product $S(\mathrm {L}) \otimes \Omega$, where $S(\mathrm {L})$ is a positive definite matrix depending on the graph Laplacian $\mathrm {L}$, and $\Omega$ is a covariance kernel on $\HH$ with appropriate regularity. In the next section, we describe appropriate choices for $S(\mathrm{ L})$ and $\Omega$.

\section{Main results}
\label{sec:main}

We first obtain the minimax rate over Sobolev balls $\mathcal{H}^{\beta,\gamma}(Q)$ for any $Q>0$. 
The minimax risk for the problem is defined by $R_n=\inf_{T_n}\sup\{ \E_f\|T_n-f\|_{2,n}^2: f\in \mathcal{H}^{\beta, \gamma}(Q)\}$, where the infimum is taken over all possible estimators. The  decay rate of the square root of the minimax risk with $n$ is called the minimax rate. 

To simplify certain bounds, we also make a simplifying assumption that $i_0=1$, and hence the smoothness condition given by the Sobolev ball can be simplified to $\sum_{i=1}^{n} \sum_{j=1}^\infty i^{2\beta/r} j^{2\gamma}\vartheta_{ij}^2 \le nQ^2$ after adjusting the constant $Q$.	We assume throughout  that $\sigma_n=1$, because the general case may be obtained by scaling: if $\epsilon_n$ is the rate obtained under this standard scaling, then the rate in the general case will be $\sigma_n \epsilon_n$.

\begin{theorem}[Minimax rate]
	\label{minimax}
	For any $Q>0$, the minimax rate for estimation of $f$ with respect to the norm $\|\cdot\|_{n}$ on  $\mathcal{H}^{\beta,\gamma}(Q)$ is $n^{-\beta \gamma /(2\beta\gamma+\beta+ r\gamma)}$. 
\end{theorem}

To derive the minimax rate, we use the canonical form of the problem. 
Our proof will use the techniques of \citet{Tmin} based on Pinsker's theorem. This requires constructing the ``Pinsker estimator'' and studying its risk, which has been addressed so far only for a single sequence. Extension to the double-index setting is a major technical advancement. 

It is natural to try and construct a Bayesian procedure attaining the minimax optimal rate. 
Recall that for a statistical model $\{X^{(n)}\sim P_\theta: \theta\in \Theta\}$ with prior $\theta\sim \Pi$, the posterior contraction rate at $\theta_0\in\Theta$ with respect to a metric $d$ is a sequence $\epsilon_n\to 0$ such that $P_{\theta_0} \Pi (\theta: d(\theta,\theta_0)>M_n \epsilon_n | X^{(n)} )\to 0$ for every $M_n\to \infty$. First we consider the case of known smoothness. 

In our setting, we assume that the covariance kernel of the functional observations is completely known. By a slight extension of our arguments, we may include an unknown scale in the formulation as long as it remains bounded between two known positive numbers and is given a positive and continuous prior density on that interval. For simplicity, we forgo the more general statements; see \citet{kirichenko2017estimating} for the additional arguments in the scalar case.

\begin{theorem}[Contraction rate: Known smoothness]
	\label{posterior known smoothness}
	Let the prior for $f$ given $c$ be $\mathrm{GP}(0,(c/n)^{(2\alpha+r)/r} (\mathrm{ L}+ n^{-2} \mathrm{I})^{-(\alpha+r/2)}\otimes \Omega)$, where $c^a\sim \mathrm{Exp}(1)$ for some $a>0$,  and $\Omega=\sum_{j=1}^\infty \kappa_j j^{-(2\gamma+1)} \psi_j\otimes \psi_j$. Then for any $Q>0$ and $f^*=(f_{1}^*,\ldots,f_{n}^*)\trans\in  \mathcal{H}^{\beta,\gamma}(Q)$ with $\beta\le \alpha+r/2$ and $\gamma \ge \max\big\{({\alpha-\beta})/{r},{\alpha}/({ar})\big\}$, the posterior contraction rate at $f^*$ with respect to $\|\cdot\|_{n}$ is $\epsilon_n= n^{-\beta \gamma /(2\beta\gamma+\beta+ r\gamma)}$. 
\end{theorem}

We note that if $\gamma\to \infty$, the underlying signal functionals become infinitely smooth, and the complexity of function estimation reduces to that of scalars or fixed dimensional objects. In this case, the posterior contraction rate coincides with $n^{-\beta/(2\beta+r)}$, the same rate \citet{kirichenko2017estimating} obtained for real-valued observations on a graph. On the other hand, if $\beta\to \infty$, the problem reduces to that of replicated ordinary functional data and the posterior contraction rate coincides with the classical estimation rate $n^{-\gamma/(2\gamma+1)}$ for $\gamma$-smooth functions. It may be noted that we have used a more general Weibull prior on $c$ instead of the exponential used by \citet{kirichenko2017estimating}, to allow a broader range of values of $\gamma$ for larger $a$. 

It is well-known  
that the posterior contraction rate is determined by the rate of concentration of the prior distribution near the true value in terms of the Kullback-Leibler divergence and the effective size of the parameter space measured by the metric entropy (\citet{ghosal2000convergence}). For Gaussian process priors, \citet{van2008rates} further showed that these properties are controlled by the RKHS of the Gaussian process prior. Let the RKHS of $\mathrm{GP}(0,\Omega)$ with $\Omega$ in Theorem~\ref{posterior known smoothness} be denoted by $\KK\subset \HH$ with RKHS norm $\|\cdot\|_{\KK}$. It is easy to see that $\KK=\{\sum_{j=1}^\infty \kappa_j^{1/2} z_j \psi_j: \sum_{j=1}^\infty j^{2\gamma+1} z_j^2<\infty \}$ with $\|\sum_{j=1}^\infty \kappa_j^{1/2} z_j \psi_j\|_{\KK}^2 = \sum_{j=1}^\infty j^{2\gamma+1} z_j^2$. Then for any fixed $c>0$, the RKHS of $\mathrm{GP}(0,(c/n)^{(2\alpha+r)/r} (\mathrm{L}+ n^{-2} \mathrm{I})^{-(\alpha+r/2)}\otimes \Omega)$, and its approximation property and small ball probabilities are characterized in the following result. The primary challenge is to simultaneously address the variation over the graph and the time domain.

\begin{lemma}
	\label{RKHS}
	The RKHS of $\mathrm{GP}(0,(c/n)^{(2\alpha+r)/r} (\mathrm{ L}+ n^{-2} \mathrm{I})^{-(\alpha+r/2)}\otimes \Omega)$ with $\Omega=\sum_{j=1}^\infty \kappa_j j^{-(2\gamma+1)} \psi_j\otimes \psi_j$ is given by $\KK^n$ with the squared RKHS norm $\|f\|^2_{\KK,c,n}= \sum_{i=1}^n (c/n)^{(2\alpha+r)/r} (\lambda_i+n^{-2})^{-(\alpha+r/2)} \|\mu_i\|_{\KK}^2$. For any $Q>0$ and $f^*\in \mathcal{H}^{\beta,\gamma}(Q)$ with $\beta\le \alpha+r/2$ and $\epsilon>0$, 
	\begin{equation}
	\label{approximation} 
	\inf\{\|h\|_{\mathbb{K},c,n}^2 : h\in \mathbb{K}^n,\|h-f^*\|_{n,2} \leq 2\epsilon\} \lesssim nc^{-(2\alpha+r)/r} \epsilon^{2-(2\alpha \gamma +r \gamma+\beta)/\beta\gamma}.
	\end{equation}
	Further, the small ball probability of the Gaussian process at $f^*\in \mathcal{H}^{\beta,\gamma}(Q)$ is estimated as 
	\begin{equation}
	\label{small ball probability}
	-\log \Pi (\|f-f^*\|\le 2\epsilon^2|c) \lesssim  c^{(2\alpha+r)/(2r\gamma)} (n\epsilon^2)^{-1/(2\gamma)} + n c^{-(2\alpha+r)/r} \epsilon^{2-(2\alpha \gamma +r \gamma+\beta)/\beta\gamma}.
	\end{equation}
\end{lemma}

We note that the prior in Theorem~\ref{posterior known smoothness} does not depend on the graphical smoothness $\beta$, and gives the targeted rate $n^{-\beta \gamma /(2\beta\gamma+\beta+ r\gamma)}$, as long as $\beta\le \alpha+r/2$. However, the prior needs to know the functional smoothness $\gamma$. The following result shows that a single prior can achieve this optimal rate (up to a logarithmic factor), that is, the Bayesian procedure adapts over the functional smoothness $\gamma$, as well as over the graphical smoothness $\beta$ on the whole range $(0,\infty)$. Here, we use a standard approach to adaptation over functional smoothness by considering a finite random series with a random number of terms (\citet{shen2015adaptive}). To obtain adaptation over the graphical smoothness, unlike \citet{kirichenko2017estimating} who used rescaled squared exponential Gaussian prior, we directly put a prior on the number of terms. The latter approach gives more numerical stability in the simulations. 

\begin{theorem}[Adaptation to smoothness]
	\label{posterior adaptation}
	Consider a finite random series prior on $f=\sum_{i=1}^{I} \sum_{j=1}^J \kappa_j^{1/2}\vartheta_{ij} \psi_j e_i$ given by $\vartheta_{ij}\sim \mathrm{N}(0,1)$, $i=1,\ldots,I$, $j=1,\ldots,J$, with $(I, J)\sim \pi$ satisfying $\exp [-a_1 I J (\log I+\log J)]\lesssim \pi(I,J) \lesssim \exp [-a_2 I J (\log I+\log J)]$ for some constants $a_1,a_2>0$. 
	Then for any $Q>0$ and $f^*\in  \mathcal{H}^{\beta,\gamma}(Q)$, the posterior contraction rate at $f_0$ with respect to $\|\cdot\|_{n}$ is $\epsilon_n= n^{-\beta \gamma /(2\beta \gamma +\beta +r\gamma)}(\log n)^{1/2+\beta \gamma /(2\beta \gamma +\beta +r\gamma) }$. 
\end{theorem}

A question of an immense interest is whether a Bayesian credible region obtained from the posterior distribution has adequate coverage. This is particularly relevant because providing a natural uncertainty quantification through the posterior distribution is an attractive feature of the Bayesian approach, and Bayesian credible sets are easy to obtain from posterior sampling. For fixed-dimensional regular families, the answer is affirmative in large samples. However, it was observed by \citet{cox}, and further clarified by \citet{knapik}, that in a smooth signal estimation problem with a white noise, under the optimal smoothing, Bayesian credible sets can have arbitrarily low coverage probabilities. This is because under the optimal smoothing, the order of the bias matches that of the variability, thus shifting a credible region from its ideal position. The problem can be alleviated by controlling the bias to a manageable level. The approach pursued in \citet{knapik} is undersmoothing, but this also makes the posterior contract sub-optimally. An alternative approach by inflating a credible region by a suitable constant (\citet{szabo},\citet{yoo}) will be pursued below. We assume that $i_0=1$ in \eqref{Laplacian eigenvalues} and the true smoothness $\beta$ and $\gamma$ are known and are not too low. We use a prior similar to Theorem~\ref{posterior known smoothness} with a deterministic $c$ and choose the specific value $\alpha=\beta-r/2$.  

\begin{theorem}[Coverage of credible ball]
	\label{coverage}
	Let the prior $\Pi$ on $f$ be a centered Gaussian process with covariance kernel given by $c^{2\beta/r} (\sum_{i=1}^{n} i^{-2\beta/r} e_i e_i\trans )\otimes \Omega$ with $\Omega=\sum_{j=1}^\infty \kappa_j j^{-(2\gamma+1)} \psi_j\otimes \psi_j$, $c=c_n= n^{(\beta+r\gamma)/(2\beta\gamma+r\gamma+\beta)}$ and $\gamma> (\beta/r)-1/2>0$. Let $q_\tau$ be the posterior $\tau$th quantile of $\|f-\hat f\|_n^2$, where $\hat f$ stands for the posterior mean. 
	Then for any $Q>0$ there exists a constant $K>0$ such that for any true function  $f^*\in  \mathcal{H}^{\beta,\gamma}(Q)$, the inflated posterior credible ball $\{f:\|f-\hat f\|_n^2 \le K q_\tau \}$ has diameter with respect to $\|\cdot\|_n$ of the order  $\epsilon_n= n^{-\beta \gamma /(2\beta\gamma+\beta+ r\gamma)}$ and its coverage $P_{f^*} (\|f^*-\hat f\|_n^2 \le K q_\tau)\to 1$. 
\end{theorem}

The crux of the proof is to show that the posterior variation around its center is of the order of the variation of the posterior mean around the truth, and both matches the optimal rate. This will be done by a careful analysis of certain bias and variability terms of both the posterior distribution and the sampling distribution of the posterior mean. Explicit expressions obtained from conjugacy are very useful in this part, and hence the value of $c$ is deterministically chosen. The coverage also holds without the assumption $\gamma>\beta/r$, by only assuming that $\beta>r/2$, but then the diameter of the credible ball has a suboptimal order. The details are omitted, but the gap can be completed using some additional estimates obtained in the proof of the next theorem.  

Now suppose that the functional data at the nodes are observed at regular grid points instead of continuously. More specifically, consider grid points $k/T$, $k=1,\ldots,T$, with meshwidth $T^{-1}$ and observations given by $Y_i(k/T)=f_i(k/T)+\varepsilon_{ik}$, where $f_1,\ldots,f_n$ are the trend functions and $\varepsilon_{ik}$ are independent $\N(0,\sigma^2)$ errors, $i=1,\ldots,n$, $k=1,\ldots,T$. These independent ``nuggets'' are thus random errors allowing fluctuations from a smooth trend of unknown functional form qualitatively related across neighboring nodes. The following result shows the posterior contraction rate for estimating the multivariate trend with respect to the normalized Euclidean distance $d_n(f,f^*)=\big\{ (nT)^{-1} \sum_{i=1}^{n} \sum_{k=1}^T |f_i(k/T)-f_i^*(k/T)|^2 \big\}^{1/2}$ on the vector of trends, where $T=T_n$ gets larger with $n$. In this case, as the errors in the functional direction are also independent, a local smoothing in that direction also allows borrowing of information and $T$ plays an essential role in the rate. We put a prior on the signal functions through a discrete wavelet transformation of the whole collection of functions $(f_1,\ldots,f_{n})$ and independent normal priors on the wavelet coefficients. After recovering the values of a function at the grid locations $k/T$, $k=1,\ldots,T$, we join the values by line-segments to construct the whole function. To simplify expressions in the proof, as in Theorem~\ref{coverage}, we assume that $i_0=1$, so that $\lambda_i\asymp (i/n)^{2/r}$ for all $i=1,\ldots,n$. The functions are assumed to satisfy a discrete version of the graph-functional joint smoothness condition. We observe an interesting phase-transition phenomenon involving the relative values of the graphical and functional smoothness --- if the individual functions are not sufficiently smooth compared with the smoothness in the graphical direction, the full benefit of the graphical smoothness may not be usable.   

\begin{theorem}[Discrete domain]
	\label{discrete}
	Let $f_i(k/T)=\sum_{m=1}^n \sum_{j=1}^T \vartheta_{mj} e_{mi}\psi_j(k/T)$, where $e_{mi}$ is the $m$th co-ordinate of the normalized eigenvector corresponding to the eigenvalue $\lambda_i$, $i=1,\ldots,n$, and $\psi_j$, $j=1,\ldots,T$, are discrete wavelet transforms. Consider a prior given by $\vartheta_{mj} \sim \N(0, (c/m)^{2\beta/r} j^{-2\gamma-1})$ independently, for some  $c\ge 1$. Let the collection of true functions $f^*=(f_1^*,\ldots,f_n^*)\trans$ with corresponding wavelet coefficients $(\vartheta_{mj}^*:m=1,\ldots,n, j=1,\ldots,T)$ satisfy the discrete smoothness condition 
	$\sum_{i=1}^{n}\sum_{j=1}^T    i^{2\beta/r} { j^{2\gamma}} (\vartheta_{ij}^*)^2 \le nTQ^2$ for some $\beta>r/2$. Then for any fixed $Q>0$, 
	\begin{enumerate}		
		\item [(i)] for $2\beta< r(2\gamma+1)$, the posterior contraction rate at $f^*$ with respect to the metric $d_n$ is given by $\epsilon_{n,T}=(nT)^{-2\beta\gamma/(4\beta\gamma+2r\gamma+r)}$ upon choosing $c=(nT)^{r(2\gamma+1)/(4\beta\gamma+2r\gamma+r)}$; 
		\item [(ii)] for $2\beta> r(2\gamma+1)$, the posterior contraction rate at $f^*$ with respect to the metric $d_n$ is given by $\epsilon_{n,T}=(nT)^{- \gamma/(2\gamma+1)}$ upon choosing $c=(nT)^{r/2\beta}$;
		\item [(iii)] for $2\beta= r(2\gamma+1)$, the posterior contraction rate at $f^*$ with respect to the metric $d_n$ is given by $\epsilon_{n,T}=(nT/\log (nT))^{- \gamma/(2\gamma+1)}$ upon choosing $c=(nT/\log (nT))^{1/(2\gamma+1)}$.
	\end{enumerate}
	Moreover, if $f_1^*,\ldots,f_n^*$ are uniformly Lipschitz continuous, then the posterior contraction rate for the full functions  $(f_1,\ldots,f_n)$ with respect to the continuous $\LL_2$-distance $\{ n^{-1}\sum_{i=1}^n \int |f_i(t)-f_i^*(t)|^2 dt\}^{1/2}$ is  $\max\{\epsilon_{n,T},T^{-1} \}$. 
\end{theorem}

The theorem clearly shows the benefit of borrowing information across neighboring nodes: if each function is individually estimated, the accuracy of estimating the trend would have been only $T^{-\gamma/(2\gamma+1)}$. This assertion is numerically supported in our simulation results and the real-data analysis by the substantially lower prediction errors for our proposed methods compared with standard prediction techniques not taking the graphical relation in consideration.

\section{Numerical illustrations}
\label{sec:simulation}

In this section, we first study the performance of the proposed method on a set of simulated data, and later apply the method to analyze daily temperatures at several weather stations. 
We consider three different graphs. We consider a sparse weighted adjacency matrix that satisfies the geometry condition approximately. To obtain a weighted adjacency matrix with 100 nodes, we generate a $50\times 50$ symmetric matrix with entries uniformly generated from $(0,1)$ and delete the edges with a weight less than $0.8$. Furthermore, we generate two more graphs with 50 nodes, the Erdos-Re\'yni random graph using R package {\tt igraph} (\citet{igraphR}) and cluster-type random graph using R package {\tt BDgraph} (\citet{BDgraphR}). These graphs also approximately satisfy the geometry condition for different values of $r$ which is pre-computed following \citet{kirichenko2017estimating}. We consider 16 equidistant timepoints in the interval $[0,1]$. Our model is thus in discreet domain having the same setup as in Theorem~\ref{discrete}. The true means is given by $f(t)=\sum_{i=0}^{n-1} 3\sin\left(\frac{i}{2T}\right)\sqrt{n}i^{-1/2-2/r}e_{i}$, where $e_i$'s are the eigenvectors of the associated graph. The mean function then satisfies the smoothness condition with $\beta=2$. Subsequently, the data is generated as $Y_{i}(t)\sim\N(f_{i}(t),1)$. We generate 50 replicated datasets for each case.

{\underline{Computational algorithm:}} 
We fit the model with prior : $\sigma^{-2}\sim$ Gamma$(0.1,0.1)$. 
Rest of the priors are as in Theorem~\ref{posterior known smoothness}.
We implement an efficient Markov chain Monte Carlo (MCMC) sampling scheme for posterior computation.
As described in Theorem\ref{discrete}, our working model $f_{i}(k/T)$ is $f_i(k/T)=\sum_{m=1}^n \sum_{j=1}^T \vartheta_{mj} e_{mi}\psi_j(k/T)$.
The error variance $\sigma^2$, and the coefficients $\vartheta_{mj}$'s are sampled using a Gibbs sampler from the full conditional conjugate posterior.
To update the scale parameter $c$, we implement a Hamiltonian Monte Carlo (HMC) sampler (\citet{neal2011mcmc}).
To select optimal $\alpha$, we consider the 5-fold cross-validation framework. To implement our method, we generate 5000 post-burn MCMC samples after burn-in 5000 samples.

We discard 50\% of data at random and train model in rest of the available data. Based on the estimated function from available data, we compute the mean at the missing locations and evaluate mean prediction MSE. In Table~\ref{result1}, we compare with the predicted values obtained  by random forest using the R package {\tt missForest} (\citet{stekhoven2012missforest}), fPCA (functional PCA) method of {\tt fdapace} (\citet{fdapaceR}) and by univariate imputation technique using the R package {\tt imputeTS} (\citet{moritz2017imputets}). The former fits a random forest on the observed part and then predicts the missing part. It reports an out-of-bag imputation error using bootstrap aggregation. The {\tt imputeTS} algorithm imputes using a spline based interpolation technique for each node independently. We use another method that uses singular value decomposition based {\tt missMDA} (\citet{missMDAR}).  Most standard imputation packages failed to produce any result. Examples include bootstrap based {\tt Amelia} (\citet{ameliaR}), Expectation-Maximization (EM) algorithm based {\tt mtsdi} (\citet{mtsdiR}), which are specially designed for spatio-temporal datasets and another EM based algorithm {\tt imputeR}. In the simulation and data application, we consider wavelet bases to construct the prior variance $\Omega$. We also compute empirical coverages of the proposed method based on equal tail 95\% credible intervals for the MCMC samples. Table~\ref{result2} contrasts empirical coverages with the coverages due to inflated credible intervals. We set the inflation factor to $0.8\log n$ following \citet{das2017bayesian}. Figure~\ref{box} compares the box plots across different methods based on the 50 simulated datasets. Spline estimates are omitted in this comparison due to their large magnitudes.

\begin{table}[ht]
	\centering
	\caption{Predictive mean square error}
	\begin{tabular}{rrrr}
		\hline
		& Weighted random graph & Erdos-Re\'yni & Cluster \\ 
		Proposed method & 1.26 & 1.20 & 1.14 \\ 
		Functional PCA & 1.43 & 1.39 & 1.34 \\ 
		Random forest & 2.12 & 1.74 & 1.79 \\ 
		PCA ({\tt missMDA}) & 2.49 & 1.73 & 1.97 \\ 
		Univariate spline imputation & 4.78 & 4.83 & 4.77 \\ 
		
		\hline
	\end{tabular}
	\label{result1}
\end{table} 

\begin{table}[ht]
	\centering
	\caption{Coverage of the proposed method for different choices of the underlying graph. Equal tail 95\% Credible Intervals (CIs) are evaluated directly from the MCMC samples. They are then inflated by a factor of $0.8\log n$ while computing inflated CIs.}
	\begin{tabular}{rrrr}
		\hline
		& Random graph & Erdos-Re\'yni & Cluster \\ 
		\hline
		Equal tail 95\% CI & 0.81 & 0.70 & 0.85 \\
		Inflated 95\% CI &0.96&0.90&0.96\\
		\hline
	\end{tabular}
	\label{result2}
\end{table} 

\begin{figure}[htbp]
	\centering
	\subfigure{\label{fig:a}\includegraphics[width = 0.3\textwidth]{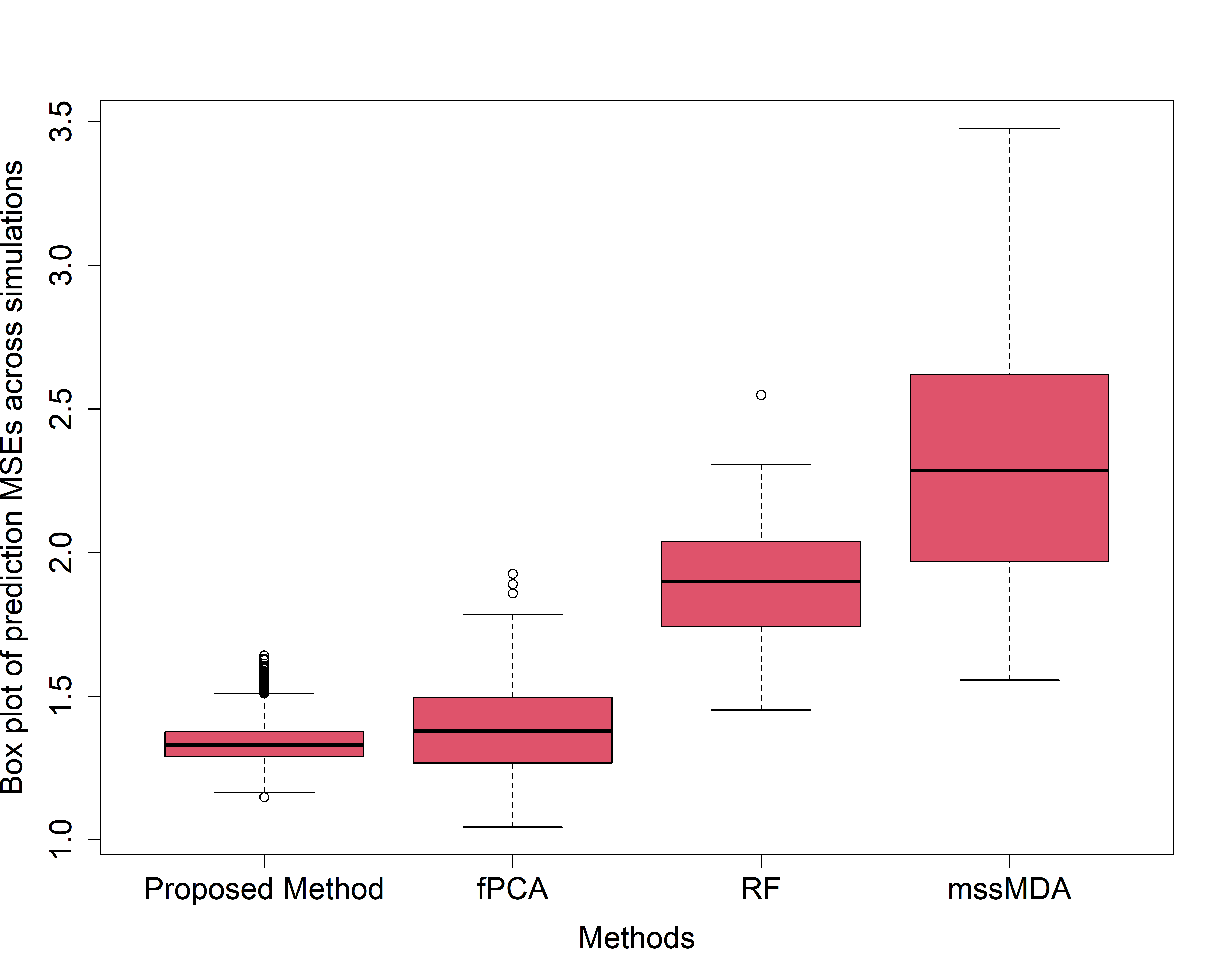}}
	\subfigure{\label{fig:a}\includegraphics[width = 0.3\textwidth]{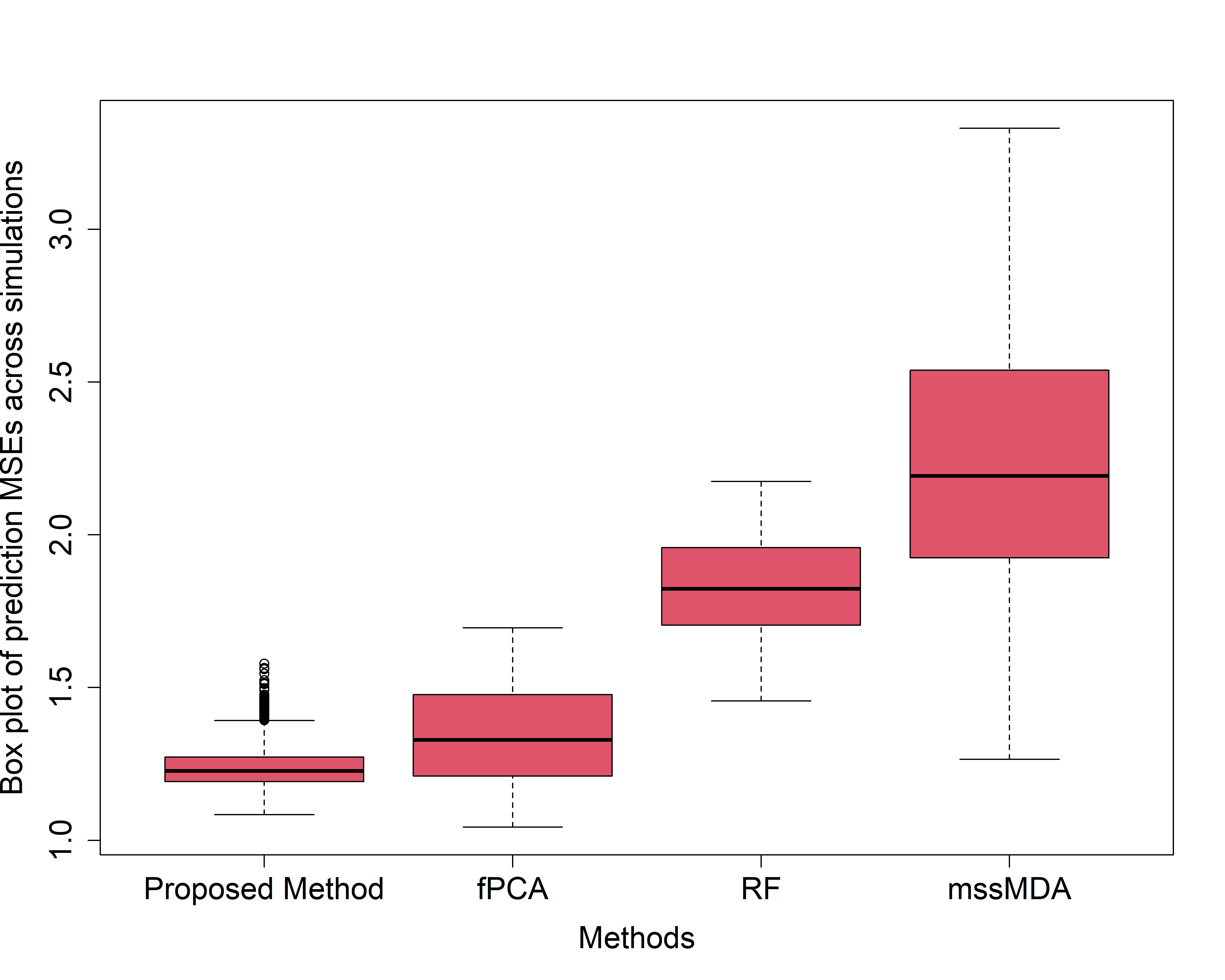}}
	\subfigure{\label{fig:a}\includegraphics[width = 0.3\textwidth]{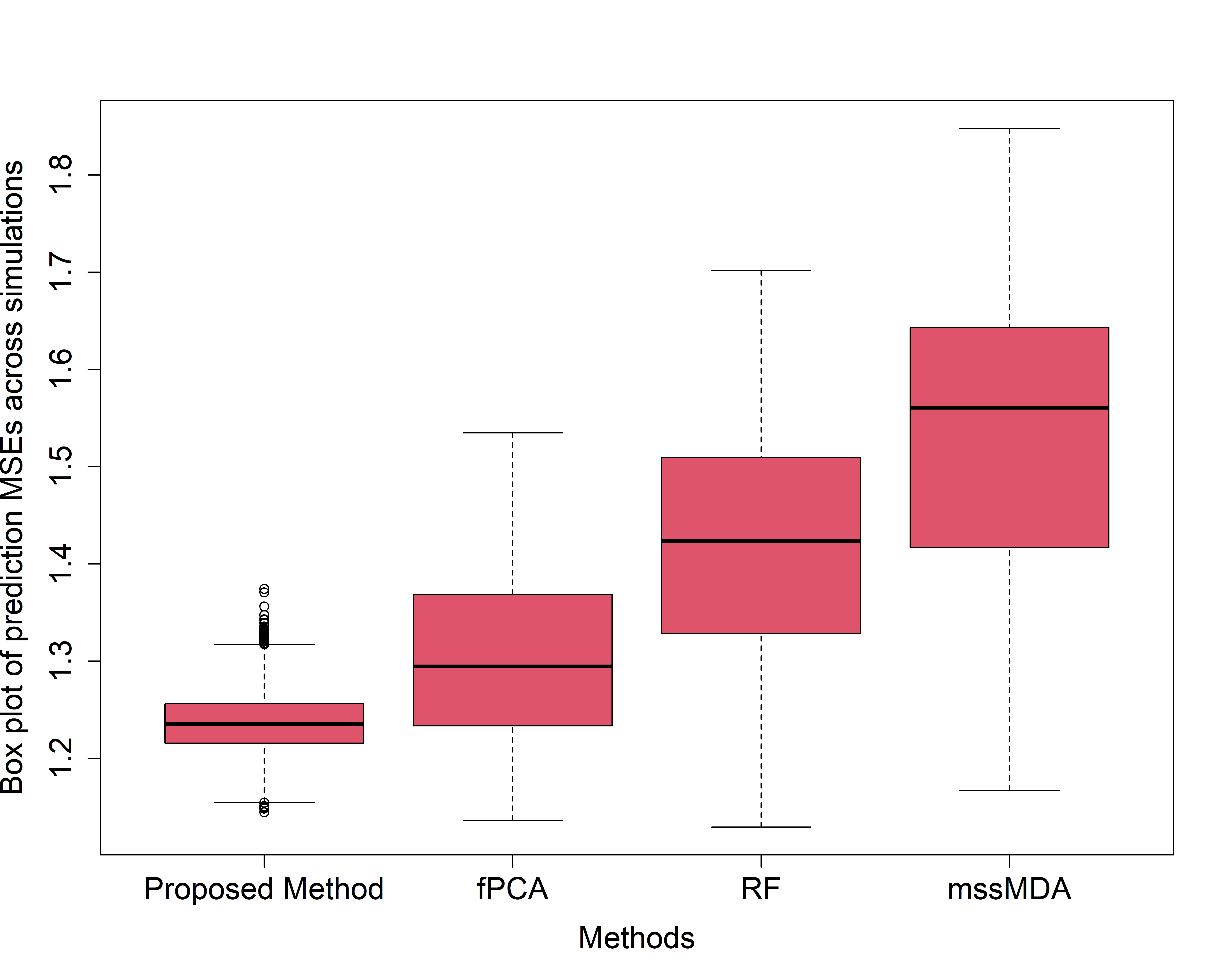}}
	\caption{Boxplots of prediction MSEs for different methods and three choices of graph structures (a) Random weighted graph;  (b) Erdos-Re\'yni  random graph; (c) Cluster graph.}
	\label{box}
\end{figure}

Now we illustrate the proposed method on a daily temperature dataset, collected by different weather stations across the US state of North Carolina. The dataset is downloaded from the website of the National Oceanic and Atmospheric Administration (NOAA). The values of daily average temperatures were collected at 158 weather locations over the year 2010. The dataset also contain the latitudes and longitudes of the weather stations. We construct weighted adjacency matrix ($A$) based on distances, as $A_{i,j}=10/d_{i,j}$. Here $d_{i,j}$ stands for the distance between $i$-$th$ and $j$-$th$ weather stations. In addition, we delete the edges having large $d_{i,j}$ to ensure sparsity. It is prudent to assume that the temperature reading of location (a) will have negligible impact on location (b) if there are geographically far apart.
In this paper, we choose $70$-$th$ percentile of $d_{i,j}$'s as threshold.
Thus, the resulting distanced based adjacency matrix in our case have $30\%$ non-zero weights after discarding the edges with large $d_{i,j}$. 
While studying evolution of mumps in England, \citet{knight2016modelling} also built a network structure among county towns based following a similar strategy.
Figure~\ref{real} shows that the constructed binary adjacency matrix satisfies the geometry condition with dimension value $r=1.96$. We again randomly  set aside $50\%$ of the observations for test data and train the model in rest. Based on the estimates, we predict at the missing locations and time points. Furthermore, we repeat this experiment using the data on four different months of the year. These four months are from four different seasons. We first mean-center and normalize the data. The comparisons here are limited to fPCA and PCA using {\tt missMDA}. Other two competing methods from the simulation section could not provide estimates.

\begin{table}[ht]
	\centering
	\begin{tabular}{rrrrr}
		\hline
		& Jan & April & July & Oct \\ 
		\hline
		Proposed method & 0.0008 & 0.0006 & 0.0005 & 0.0007 \\ 
		Functional PCA & 0.0070 & 0.0060 & 0.0070 & 0.0080 \\ 
		PCA ({\tt missMDA}) & 0.3040 & 0.3300 & 0.3600 & 0.1200 \\ 
		\hline
	\end{tabular}
\end{table}

\begin{figure}[h] 
	\centering
	\subfigure{\includegraphics[width=5cm]{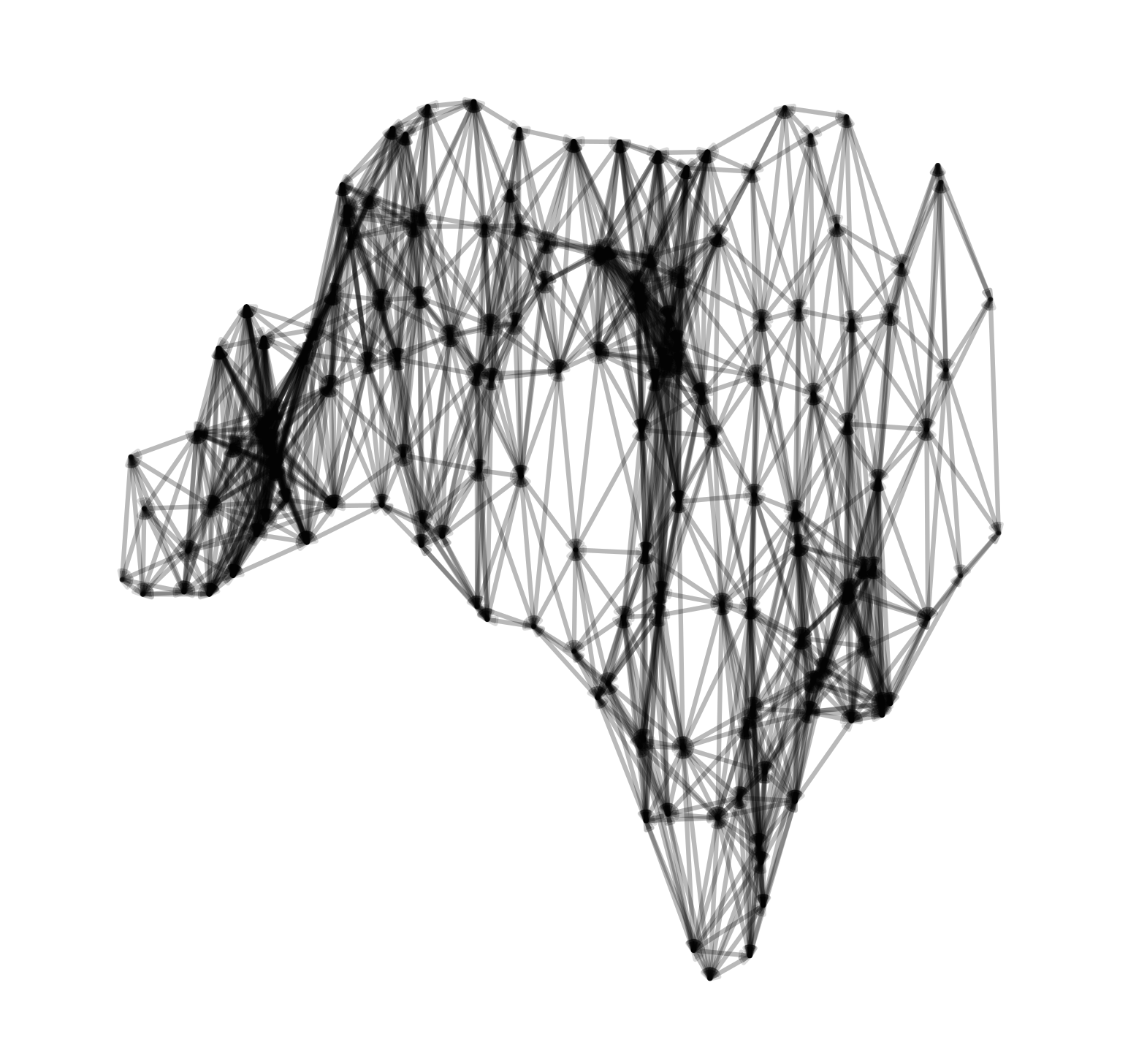}} \hspace{2cm}
	\subfigure{\includegraphics[width = 5cm]{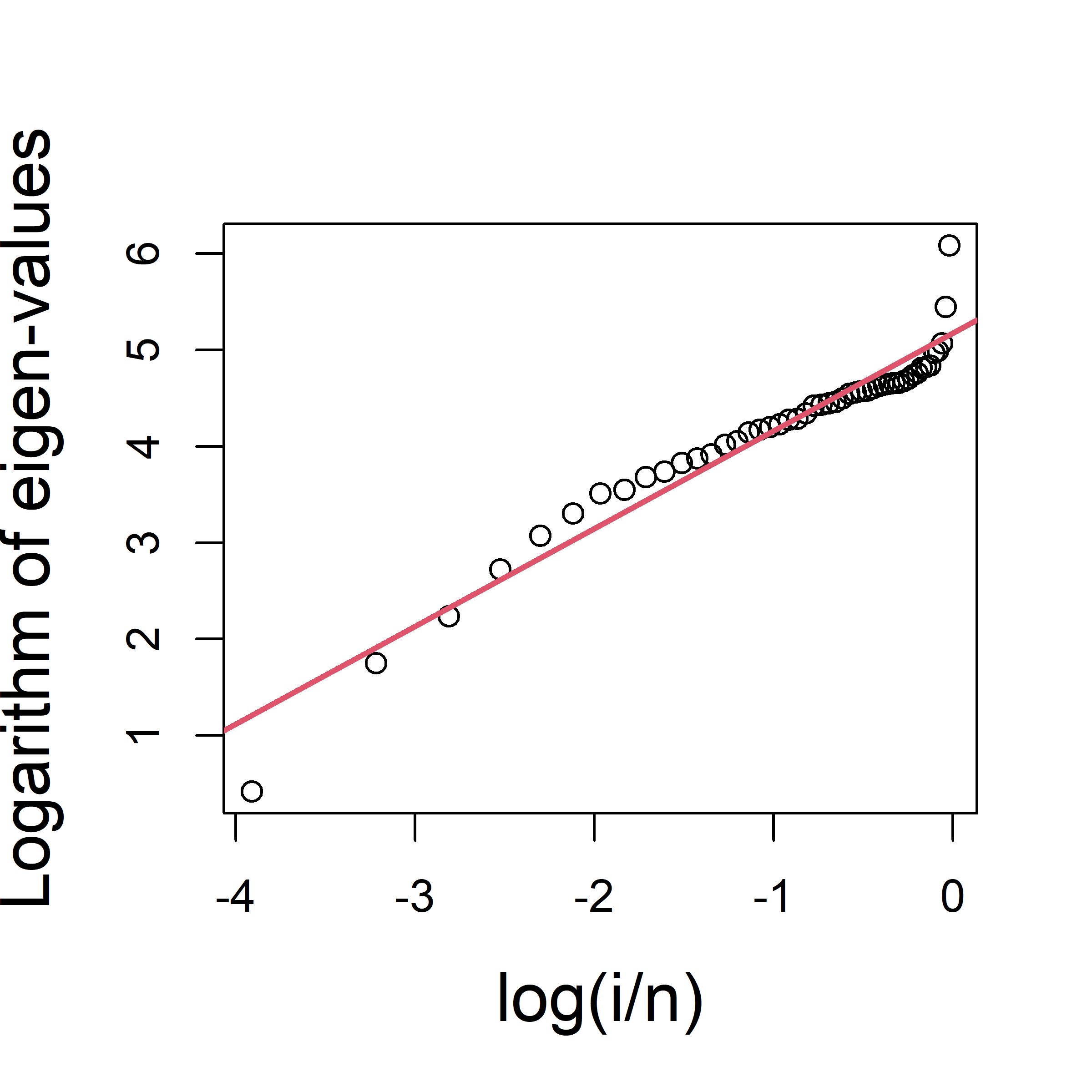}}
	\caption{ (a) The graph structure of the weather stations network;  (b) Logarithm of the $i$th eigenvalues of Laplacian plotted against $\log(i/n)$ to determine the graphical dimension.}
	\label{real}
\end{figure}

\section{Discussion}
The statistical methods to analyze functional data have observed many new developments in the recent past (see \citet{dabo2008functional,ferraty2011recent, goia2016special, aneiros2017functional, aneiros2019editorial}). 
Functional PCA based approaches have been incredibly successful to model functional datasets in dense grid \citet{hall2006properties}. In Bayesian framework, clustering based approaches are proposed for functional clustering \citep{petrone2009hybrid, rodriguez2009bayesian}, and also in image regression \citep{meyer2015bayesian, goldsmith2014smooth}. \cite{wang2016functional} provides a thorough review on recent advancements on analyzing functional data.

In this work, we develop a method to study functional data on a given network. 
To the best of our knowledge, this is the first attempt on a network-linked FDA. 
The main novelties, apart from considering an FDA setup, is that we address the two-way smoothness issue, for studying minimax rate and Bayesian adaptation, a pioneering study, and we establish coverage in the graphical setting, which has not been done earlier even for scalar observations. 
There are some immediate extensions, we may consider.
One extension could be modeling a multivariate functional dataset taking a two-stage approach. The graphical dependence may be computed in the first stage. The subsequent stage may apply our proposed functional data model using the estimated network from the first stage.
Another extension may be to consider modeling time-varying networks.

\section*{Acknowledgement}
We would like to thank the Editors of the journal, the editors of special issue, and reviewers for their constructive comments that improved presentation of the paper. The second author is partially funded by the ARO grant 76643-MA.

\section*{Appendix: Proofs}
\label{sec:proofs}

\begin{proof}[Proof of Theorem~\ref{minimax}]
	We consider the problem reduced in the canonical form in terms of observation $Z_{ij}$ and parameter $\vartheta_{ij}$. We identify $\vartheta$ with $f$ and say $\vartheta\in \mathcal{H}^{\beta, \gamma}(Q)$ if $ \sum_{i=1}^{n} \sum_{j=1}^\infty i^{2\beta/r}j^{2\gamma} \vartheta_{ij}^2\le nQ^2$. Also $\|f\|_n^2$ reduces to  $n^{-1} \sum_{i=1}^{n}\sum_{j=1}^\infty \vartheta_{ij}^2$. 
	
	We follow Pinsker's approach as described in \citet{Tmin} suitably adapted to double arrays. Consider a linear estimator $l(Z):=(\!(l_{ij} Z_{ij}\!)$ of $\vartheta:=(\!(\vartheta_{ij})\!)$ and compute its (normalized) risk $R(l,\vartheta)=n^{-1}\sum_{i=1}^{n}\sum_{j=1}^{\infty}\{ (1-l_{ij})^2\vartheta_{ij}^2+  l_{ij}^2\}$. 
	Define Pinsker's estimator by taking the coefficients to be $l_{ij}^*=(1-\delta b_{ij})_{+}$, where $b_{ij}=i^{\beta/r} j^\gamma$ and $\delta$ is the solution of 
	$\delta^{-1} {n^{-1}}\sum_{i=1}^{n}\sum_{j=1}^{\infty}b_{ij}(1-\delta b_{ij})=Q^2$. 
	As $b_{ij}\geq 0$, is increasing in $i$ and $j$, and $b_{ij}\rightarrow \infty$ if either $i\to \infty$ or $j\to \infty$, following the argument given in 
	Lemma~3.1 of \citet{Tmin}, we conclude that $\delta$ is unique and is given by $\delta = {n^{-1}\sum_{(i,j)\in \mathcal{D}}b_{ij}}/\{Q^2+n^{-1}\sum_{(i,j)\in \mathcal{D}}b_{ij}^2\}$, where 
	$$\mathcal{D}=\big\{(i, j): b_{ij} \le  1/\delta \big\}=\big\{(i, j): 1\leq j\leq\delta^{-1/\gamma}i^{-\beta/r\gamma}, 1\leq i\leq\delta^{-r/\beta}\big\}$$ 
	has cardinality $\#\mathcal{D}\le\sum_{i=1}^{\delta^{-r/\beta}}\delta^{-1/\gamma}i^{-\beta/r\gamma}\lesssim \delta^{-1/\gamma-r/\beta}$. Now 
	$$
	nQ^2=\delta^{-1} \sum_{i=1}^{n}\sum_{j=1}^{\infty}b_{ij}(1-\delta b_{ij})_{+} \asymp \delta^{-1}\sum_{(i,j)\in \mathcal{D}} i^{\beta/r}j^{\gamma} 
	\asymp  \delta^{-2}\#\mathcal{D}\leq \delta^{-2-1/\gamma-r/\beta}
	$$			
	since $i^{\beta/r}j^{\gamma}\le 1/{\delta}$ for $(i,j)\in \mathcal{D}$, giving that  $\delta\asymp n^{-\gamma\beta/(2\gamma\beta+\beta+r\gamma)}$. 
	
	We estimate the upper bound of the minimax risk by the risk  of Pinsker's estimator  $l^*(Z)$ given by  
	\begin{align*}
	R(l^*(Z),\vartheta) 
	& =  n^{-1} \big[\delta^2 \sum_{(i,j)\in \mathcal{D}} b_{ij}^2 \vartheta_{ij}^2 + \sum_{(i,j)\in \mathcal{D}^c} \vartheta_{ij}^2+ \sum_{(i,j)\in \mathcal{D}} (1-\delta b_{ij})^2\big].
	\end{align*}
	The first term is bounded by $\delta^2 Q^2$. By the definitions of $\mathcal{H}^{\beta,\gamma}(Q)$ and $\mathcal{D}$, the second by $n^{-1} \max \{ i^{-2\beta/r}j^{-2\gamma}: (i,j)\in {\mathcal{D}}\} \sum_{i=1}^{n} \sum_{j=1}^\infty  i^{2\beta/r} j^{2\gamma} \vartheta_{ij}^2\le \delta^2 Q^2$. The last term is $n^{-1} \sum_{(i,j)\in \mathcal{D}} (1-\delta b_{ij})_+^2\le n^{-1} \# \mathcal{D}\le n^{-1} \delta^{-1/\gamma-r/\beta}\asymp \delta^2 \asymp n^{-2\beta\gamma/(2\beta\gamma+\beta+r\gamma)}$.  Thus the upper bound follows.

	We lower bound the minimax risk by a Bayes risk. Let  $\vartheta_{\mathcal{D}}=(\vartheta_{ij}: (i,j)\in {\mathcal{D}})$ and $\Theta_{\mathcal{D}}=\{\vartheta_{\mathcal{D}}: \sum_{(i,j)\in {\mathcal{D}}}b_{ij}^2\vartheta_{ij}^2\leq nQ^2\}$. Let a prior  $\Pi$ for  $\vartheta_{\mathcal{D}}$ be $\vartheta_{ij}\sim \mathrm{N}(0, (1-\delta)v_{ij}^2)$ independently, where $v_{ij}^2=(1-\delta b_{ij})_+/(\delta b_{ij})$. By conjugacy
	$\vartheta_{ij}|Z_{ij}\sim \N ((1-\delta) Z_{ij}/(1-\delta+ v_{ij}^{-2}), (1-\delta)/(1-\delta+ v_{ij}^{-2})))$. 
	Then a lower bound to the minimax risk is given by 
	$M^*-m^*$, where $M^*$ is the minimal Bayes risk with respect to the prior $\Pi$ above   
	and $m^*$ is the maximum Bayes risk with respect to $\Pi$ for an estimator lying in $\Theta_D$.   Thus  it follows by simple calculations that 
	$$M^*=\frac{\delta (1-\delta)}{n} \sum_{(i,j) \in {\mathcal{D}}} \frac{b_{ij}(1-\delta b_{ij})}{1-\delta+\delta^2 b_{ij}}+\frac{1-\delta}{n} \sum_{(i,j) \in {\mathcal{D}}} (1-\delta b_{ij}) \gtrsim \frac{\delta}{n}\sum_{(i,j) \in {\mathcal{D}}}  b_{ij} +\frac{\# {\mathcal{D}}}{n}, $$ 
	which is of the order $n^{-1}\delta^{-1/\gamma-r/\beta}\asymp \delta^2$, 
	the same order as the upper bound, since $\delta\asymp n^{-\beta\gamma/(2\beta\gamma+\beta+r\gamma)}\to 0$. It then remains to verify that $m^*$ is negligible compared to $M^*$. 
	
	The maximum Euclidean norm $\|\vartheta_{\mathcal{D}}\|$ of $\vartheta_{\mathcal{D}}\in  \Theta_{\mathcal{D}}$ is clearly at most a multiple of $\sqrt{n}Q$.  Then by following the arguments in  pages 150--154 of \citet{Tmin}, we get 
	$$m^*\leq 2 [\sup\{\|\vartheta_{\mathcal{D}}\|^2: \vartheta_{\mathcal{D}}\in\Theta_{\mathcal{D}}\}  \Pi(\Theta_{\mathcal{D}}^c)+\{\Pi(\Theta_{\mathcal{D}}^c)
	\E_{\Pi} \|\vartheta_{\mathcal{D}}\|^4 \}^{1/2}] \lesssim    n \{\Pi(\Theta_{\mathcal{D}}^c)\}^{1/2}$$ 
	using a bound for the fourth moment of a centered normal variable. 
	It then suffices to show that $\Pi(\Theta_{\mathcal{D}}^c) $ is exponentially small in $n$, since $M^*$ decays as a power of $n$. Write  $\Pi(\Theta_{\mathcal{D}}^c)=\P(\sum_{(i,j)\in \mathcal{D}}b_{ij}^2 v_{ij}^2 \eta_{ij}^2>u)$, where $u=nQ^2/(1-\delta)\asymp n$ and $\eta_{ij}$ are independent standard normal variables. 
	By Markov's inequality, 
	\begin{align} 
	\P\big(\sum_{(i,j)\in \mathcal{D}}b_{ij}^2 v_{ij}^2 \eta_{ij}^2>u\big)\leq e^{-qu}\prod_{(i,j)\in \mathcal{D}}\E\{\exp(q b_{ij}^2 v_{ij}^2\eta_{ij}^2 )\}=e^{-qu}\prod_{(i,j)\in \mathcal{D}}(1-2q b_{ij}^2 v_{ij}^2)^{-1/2} .
	\label{neweqn}
	\end{align}
	Using the facts that $(1-2x)^{-1/2}\leq e^{2x}$ for $x < 1/4$, $\max \{ qb_{ij}^2 v_{ij}^2: (i,j)\in \mathcal{D} \}\le  q/\delta^2<1/4$ for $q<\delta^2/4$, we can bound the right hand side of \eqref{neweqn} by 
	$\exp\{-qu+4q\sum_{(i,j)\in \mathcal{D}} b_{ij}^2 v_{ij}^2 \} $. 	
	Now, $q\sum_{(i,j)\in \mathcal{D}} b_{ij}^2 v_{ij}^2 \leq  q\delta^{-1} \sum_{(i,j)\in \mathcal{D}} b_{ij} \le q\delta^{-2} \# \mathcal{D} \asymp qn \lesssim n\delta^2$. Hence  for large enough $u$, which we can have when $Q$ is large, it follows that the bound in \eqref{neweqn} decays exponentially in $n\delta^2\to \infty$.  
\end{proof}

\begin{proof}[Proof of Lemma~\ref{RKHS}]
	We first lower bound the small ball probability $\Pi (\|f\|_n^2 \le \epsilon^2|c) $ at the origin. We can express the function as 
	$	f=\sum_{i=1}^{n} \sum_{j=1}^\infty \kappa_j^{1/2} j^{-(\gamma+1/2)} (c/n)^{(2\alpha+r)/2r}(\lambda_i+n^{-2})^{-(\alpha+r/2)/2} W_{ij} \psi_j e_i$, where $W_{ij}$ are independent standard normal variables.
	We split the sum above in four regions: $(i< i_0, j\le j_1)$, $(i< i_0, j>j_1)$, $(i\ge i_0,j\le j_2)$ and $(i\ge i_0, j\le j_2)$, where $j_1,j_2$ depend on $n$ and will be specified later. 
	From \eqref{Laplacian eigenvalues}, we have the estimates $(\lambda_i+n^{-2})^{-(\alpha+r/2)}\le C_3 n^{-(2\alpha+r)}$ for all $i$ and an improved estimate that gives $(\lambda_i+n^{-2})^{-(\alpha+r/2)}\le C_4 (i/n)^{-(2\alpha+r)/r}$ for $i\ge i_0$, where $C_3=(C_1+1)^{-(\alpha+r/2)}$ and $C_4=C_1^{-(\alpha+r/2)}$. 
	Then $\|f\|_n^2=n^{-1} \sum_{i=1}^{n} \sum_{j=1}^\infty j^{-(2\gamma+1)} (c/n)^{(2\alpha+r)/r}(\lambda_i+n^{-2})^{-(\alpha+r/2)} W_{ij}^2$ can be split in sums over these regions and it suffices to lower bound the probabilities 
	\begin{itemize}
		\item [(i)] $\P (C_3 \sum_{i< i_0} \sum_{j\le j_1}  j^{-(2\gamma+1)} (c/n)^{(2\alpha+r)/r}n^{-(2\alpha+r)} W_{ij}^2 \le n\epsilon^2/4)$; 
		\item [(ii)] $\P (C_3 \sum_{i< i_0} \sum_{j> j_1}  j^{-(2\gamma+1)} (c/n)^{(2\alpha+r)/r}n^{-(2\alpha+r)} W_{ij}^2 \le n\epsilon^2/4)$; 
		\item [(iii)] $\P (C_4 \sum_{i\ge  i_0} \sum_{j\le j_2}  j^{-(2\gamma+1)} (c/n)^{(2\alpha+r)/r} (i/n)^{-(2\alpha+r)/r} W_{ij}^2 \le n\epsilon^2/4)$; 
		\item [(iv)] $\P (C_4 \sum_{i\ge  i_0} \sum_{j> j_2}  j^{-(2\gamma+1)} (c/n)^{(2\alpha+r)/r} (i/n)^{-(2\alpha+r)/r} W_{ij}^2 \le n\epsilon^2/4)$.
	\end{itemize}
	We shall show that (ii) and (iv) are greater than $1/2$ for appropriate choices of $j_1,j_2$ and then estimate (i) and (iii) for those choices. 
	
	Note that the expected value of the expression in (ii) is $C_3 c^{(2\alpha+r)/r}  i_0 \sum_{j> j_1}  j^{-(2\gamma+1)}\le C_5 c^{(2\alpha+r)/r} (j_1+1)^{-2\gamma}$, where $C_5=C_3 i_0/(2\gamma)$. Hence by Markov's inequality, (ii) is at least $1/2$ if $j_1$ is chosen to be the integer part of $C_6 c^{(2\alpha+r)/(2r\gamma)}(n\epsilon^2)^{-1/(2\gamma)}$, where $C_6=(8C_5)^{1/(2\gamma)}$. 
	
	Now we can lower bound the probability in (i) by $[\P (\sum_{j\le j_1}  j^{-(2\gamma+1)}  W_{1j}^2 \le \delta^2)]^{i_0}$, where $\delta^2= c^{-(2\alpha+r)/r}n\epsilon^2/(4C_3 i_0)$. From Lemma~6.2 of \citet{belitser2003adaptive}, $$\P (\sum_{j\le j_1}  j^{-(2\gamma+1)}  W_{1j}^2 \le \delta^2)\ge a_1^{-j_1}  \P (\sum_{j\le j_1} W_{1j}^2 \le 2 \delta^2 j_1^{2\gamma+1} )\ge a_1^{-j_1} [\P ( W_{11}^2 \le 2 \delta^2 j_1^{2\gamma})]^{j_1},$$ 
	where $a_1=\sqrt{2} e^{2\gamma+1}$. Plugging in the value of $j_1$, it follows that $2 \delta^2 j_1^{2\gamma}=4C_5/i_0$, and hence the estimate reduces to $a_2^{-j_1}$ for another constant $a_2>0$. 
	
	The expected value of the sum in (iv) is $C_4 \sum_{i\ge i_0} \sum_{j> j_2}  j^{-(2\gamma+1)} c^{(2\alpha+r)/r} i^{-(2\alpha+r)/r}\le C_7 c^{(2\alpha+r)/r} (j_2+1)^{-2\gamma}$, where $C_7=C_4 i_0^{-2\alpha/r}/(2\gamma)$. Thus the probability in (iv) is at least $1/2$ if $j_2$ is chosen to be the integer part of $C_8 c^{(2\alpha+r)/(2r\gamma)}(n\epsilon^2)^{-1/(2\gamma)}$, where $C_8=(4 C_7)^{1/(4 \alpha\gamma)}$. 
	The probability in (iii) is
	$\P ( \sum_{i> i_0} \sum_{j\le j_2} i^{-(2\alpha+r)/r} W_{ij}^2 \le \tau^2  j_2^{2\gamma+1}) \ge [\P ( \sum_{i> i_0} i^{-(2\alpha+r)/r} W_{i1}^2 \le \bar \delta^2  j_2^{2\gamma}) )]^{j_2}$, 
	where $ \bar\delta^2 = (4C_4)^{-1} c^{-(2\alpha+r)/r} n\epsilon^2$. With $j_2$ chosen as above, $\bar \delta^2  j_2^{2\gamma}$ is a fixed constant $C_9$. By Corollary~4.3 of \citet{dunker1998small}, $\P ( \sum_{i\ge i_0} i^{-(2\alpha+r)/r} W_{i1}^2 \le C_9)\ge \exp(-C_{10} C_9^{-r/\alpha})=a_3$, where $C_{10}$ and $a_3$ are positive constants.  Hence the probability in (iii) is lower bounded by $a_3^{j_2}$.  
	
	Now combining all estimates, it follows that $\Pi (\|f\|_n^2 \le \epsilon^2|c) \ge a_2^{-j_1} a_3^{j_2}/4$, or that 
	\begin{equation}
	\label{log small ball}
	-\log \Pi (\|f\|_n^2 \le \epsilon^2|c)\le j_1 \log a_2+ j_2 \log a_3+\log 4\lesssim c^{(2\alpha+r)/(2r\gamma)} (n\epsilon^2)^{-1/(2\gamma)}.
	\end{equation}

	We now show the result on decentering function at an $f^*\in \mathcal{H}^{\beta,\gamma}(Q)$ for some $Q>0$. We can represent $f^*=\sum_{i=1}^{n} \sum_{j=1}^\infty \kappa_j^{1/2} \vartheta_{ij} \psi_j e_i$  with $n^{-1}\sum_{i=1}^{n} \sum_{j=1}^\infty j^{2\gamma} (1+n^{2\beta/r} \lambda_i^\beta)\vartheta_{ij}^2\le Q^2$. Let $h=\sum_{i=1}^{I} \sum_{j=1}^J\kappa_j^{1/2} \vartheta_{ij} \psi_j e_i$, with $I$ and $J$ to be chosen below. Then the residual squared-norm $\|h-f_0\|_n^2$ is bounded by $\sum_{i,j: \,i>  I} \vartheta_{ij}^2 +\sum_{i,j:\, j >  J} \vartheta_{ij}^2$. Because $\lambda_i \le C_2 (i/n)^{2/r}$ for all $i\ge i_0$ and $\lambda_i$ are increasing, this can be further bounded by 
	$$ I^{-2\beta/r} n^{-1} \sum_{i=1}^n \sum_{j=1}^\infty j^{2\gamma} (1+n^{2\beta/r} \lambda_i^\beta)\vartheta_{ij}^2+J^{-2\gamma} n^{-1} \sum_{i=1}^n \sum_{j=1}^\infty j^{2\gamma} (1+n^{2\beta/r} \lambda_i^\beta)\vartheta_{ij}^2\lesssim I^{-2\beta/r}+  J^{-2\gamma}.$$
	In order to bound this by $\epsilon^2$, we choose $I\asymp \epsilon^{-r/\beta}$ and $J\asymp \epsilon^{-1/\gamma}$. The smoothness condition implies that for $i_0\le I\le  n$, we have $\sum_{i=1}^{I}\sum_{j=1}^J (1+i^{2\beta/r}) j^{2\gamma}\vartheta_{ij}^2\leq nQ^2$. Since the eigenvalues of the covariance kernel of $\mathrm{GP}(0,(c/n)^{(2\alpha+r)/r} (\mathrm{L}+ n^{-2}\mathrm{ I})^{-(\alpha+r/2)}\otimes \Omega)$ corresponding to the eigenfunctions $\psi_j e_i$ are $(c/n)^{(2\alpha+r)/2r}(\lambda_i+n^{-2})^{-(\alpha +r/2)/2}  j^{-(\gamma+1/2)}$, it follows that the squared RKHS norm of $h\in \KK^n$ given by $\|h\|_{\KK,c,n}^2=\sum_{i\le I} \sum_{j \le J}  (n/c)^{(2\alpha+r)/r} j^{2\gamma+1} (\lambda_i+n^{-2} )^{\alpha +r/2} \vartheta_{ij}^2$ is bounded by a constant multiple of 
	\begin{align*}
	c^{-(2\alpha+r)/r}\sum_{i\leq I}\sum_{j \leq J}j i^{(2\alpha+r-2\beta)/r} (j^{2\gamma}i^{2\beta/r})\vartheta_{ij}^2 \lesssim nc^{-(2\alpha+r)/r}JI^{(2\alpha+r-2\beta)/r} Q^2.
	\end{align*}
	By the choice of $I$ and $J$, this is bounded by a constant multiple of   	 $nc^{-(2\alpha+r)/r}\epsilon^{-(({2\alpha+r-2\beta})/{\beta})-{1}/{\gamma}}$.
	
	Combining with the estimate of the small ball probability at the origin obtained above, the assertion on the prior probability of concentration at the true function follows (see, e.g., Proposition~11.19 of \citet{Ghosal}). 
\end{proof}

\begin{proof}[Proof of Theorem~\ref{posterior known smoothness}]
	We intend to show that the posterior contracts at the rate $\epsilon_n=n^{-\beta \gamma /(2\beta \gamma+\beta +r\gamma)}$. 	
	
	To apply the general theory for posterior contraction (\citet{Ghosal}) with respect to the norm $\|\cdot\|_{n}$, we first establish, for a pair of functions $f^*$ and $f^\dagger$ with  $\|f^\dagger-f^*\|_{n}>\epsilon$, the existence of tests for $f=f^*$ against $\{ f: \|f-f^\dagger\|_{n}<\epsilon/4\}$ with error probabilities bounded by $e^{- n\epsilon^2/32}$; see (8.17) of \citet{Ghosal}. We can show that the likelihood ratio test for testing $f=f_0$ against $f=f_1$ satisfies the requirement. The proof proceeds by applying Lemma~D.16 of \citet{Ghosal} on the canonical model in terms of the independent variables $Z_{ij}\sim \N(\vartheta_{ij},1)$. Whenever $n^{-1}\sum_{i=1}^n \sum_{j=1}^\infty  |\vartheta_{ij}^*-\vartheta_{ij}^\dagger|^2\ge \epsilon^2$, there exists a test $\phi_n$ for testing $\vartheta_{ij}=\vartheta_{ij}^*$ for all $i,j$ against $n^{-1} \sum_{i=1}^n \sum_{j=1}^\infty  |\vartheta_{ij}-\vartheta_{ij}^\dagger|^2\le \epsilon^2/16$ with both type of error probabilities uniformly bounded by $e^{-n\epsilon^2/32}$. In terms of the original function $f$, this translates to a test sequence for $f=f^*$ against $\{ f: \|f-f^\dagger\|_{n}<\epsilon/4\}$ with both type of error probabilities uniformly bounded by $e^{-n\epsilon^2/32}$, as required. 
	
	Next, we need to verify a prior concentration property in Kullback-Leibler neighborhoods given by (8.19) of \citet{Ghosal}. By considering the equivalent model, we can apply Lemma~L.4 of \citet{Ghosal} that the required Kullback-Leibler neighborhood is equivalent with $\|\cdot\|_{n}$-neighborhood, so it suffices to show the estimate that $-\log \Pi (f: \|f-f^*\|_{n}\le \epsilon_n)\lesssim n\epsilon_n^2=n^{(\beta+r\gamma)/(2\beta \gamma+\beta +r\gamma)}$. Clearly, $$\Pi (f: \|f-f^*\|_{n}\le \epsilon_n)
	\ge (e^{-c_{1n}^a}-e^{-c_{2n}^a}) \inf \{\Pi (f: \|f-f^*\|_{n}\le \epsilon_n|c): {c_{1n}}\le c\le {c_{2n}} \},$$ where we choose $c_{1n}=n^{2\gamma r (\beta-\alpha)(2\alpha+r)(2\beta \gamma+\beta+r\gamma)}$ and $c_{2n}=n^{r(\beta+r\gamma)(1+2\gamma)/(2\alpha+r)(2\beta\gamma+\beta+r\gamma)}$. We observe that the interval is not vacuous, that is, $c_{1n}\le c_{2n}$. This is clearly true if $\beta\le \alpha$, while for $\beta\ge \alpha$, this can be seen to hold after some simplification using the assumed bound $2(\beta-\alpha)\le r$.  Note that $c_{2n}$ and $c_{1n}$ are respectively the solutions of $ c^{(2\alpha+r)/(2r\gamma)} (n\epsilon_n^2)^{-1/(2\gamma)}\le n\epsilon_n^2$ and 
	$c^{-(2\alpha+r)/r} \epsilon_n^{-(2\alpha \gamma +r \gamma+2\beta-2\beta\gamma)/(2\beta\gamma)}\le n\epsilon_n^2$, implying that  $-\log \Pi (f: \|f-f^*\|_{n}\le \epsilon_n|c)\lesssim n \epsilon_n^2$ for all $c_{1n}\le c \le c_{2n}$. Further, by our choice $c_{1n}\le n\epsilon_n^2$, so that the factor $e^{-c_{1n}^a}-e^{-c_{2n}^a}$ can be absorbed in the second factor leading to  $-\log \Pi (f: \|f-f^*\|_{n}\le \epsilon_n)\lesssim n\epsilon_n^2$, as required for the contraction rate $\epsilon_n$. 
	
	To complete the proof, it remains to construct a sieve $\mathcal{S}_n\subset \HH^n$ such that $\Pi (\mathcal{S}_n^c)\le e^{-A n \epsilon_n^2}$ with a sufficiently large $A>0$ and $\log N(\epsilon_n, \mathcal{S}_n, \|\cdot\|_n)\lesssim n\epsilon_n^2$, where $N$ stands for the covering number. 
	Let $ \HH^n_1$ stand for the unit ball in $\HH^n$ and $ \KK_{c,1}^n$. 
	By Borell's inequality, $\Pi (f \not \in \epsilon_n \HH^n_1+  C\sqrt{n}\epsilon_n \KK_{c,1}^n|c)\le e^{-A n \epsilon_n^2}/2$ if $C>0$ is chosen sufficiently large. If we choose $\mathcal{S}_n= \epsilon_n \HH^n_1+  C\sqrt{n}\epsilon_n \KK_{c_{2n},1}^n$, then we have $\Pi (\mathcal{S}_n^c)\le \Pi (c>c_{2n})+\int_0^{c_{2n}}  \Pi ( f \not \in \epsilon_n \HH^n_1+  C\sqrt{n}\epsilon_n \KK_{c,1}^n|c)   d\pi (c)\le e^{-c_{2n}^a}+ \frac12 e^{-A n \epsilon_n^2}$ because $\|\cdot\|_{\KK,c,n}$ is monotone decreasing in $c$, and hence the unit ball $\KK_{c,1}^n$ is monotone increasing in $c$. Thus  $\cup_{ c\le c_{2n}} \KK_{c,1}^n =\KK_{c_{2n},1}^n$ and the required condition on $\Pi (\mathcal{S}_n^c)$ is satisfied in view of the fact that $c_{2n}\gg n\epsilon_n^2$, which is a consequence of $\gamma\ge \alpha/ar$. 
	
	Finally, the metric entropy $\log N(2\epsilon_n, \mathcal{S}_n, \|\cdot\|_n)\le \log N (\epsilon_n, C\sqrt{n}\epsilon_n \KK_{c_{2n},1}^n , \|\cdot\|_n)$, and to compute the latter, we approximate any $f=\sum_{i=1}^{n}\sum_{j=1}^\infty \kappa_j^{1/2} j^{-(\gamma+1/2)} \vartheta_{ij} \psi_j e_i\in \KK^n$ by $f_{I,J}=\sum_{i=1}^{I}\sum_{j=1}^{J} \kappa_j^{1/2} j^{-(\gamma+1/2)} \vartheta_{ij} \psi_j e_i\in \KK^n$ with $I\asymp \epsilon_n^{-r/(2\beta)}$ and $J\asymp \epsilon_n^{-1/(2\gamma)}$. We observe from the proof of Lemma~\ref{RKHS} that the approximation error is bounded by $c_{2n}^{-(2\alpha+r)/r} \epsilon_n^{2-(2\alpha \gamma+r\gamma+2\beta)/(2\beta\gamma)}=n \epsilon_n^2$ for our choice of $c_{2n}$. Hence the computation of the $\epsilon_n$-metric entropy reduces to that of an  $IJ$-dimensional $C\epsilon_n$-ball with respect to the Euclidean metric. The latter equals, up to a constant multiple, the dimension $IJ\asymp \epsilon_n^{-r/(2\beta)-1/(2\gamma)}$, which is easily verified to be equal to  $n\epsilon_n^2$.  
	
	Piecing the estimates together and applying Theorem~8.23 of \citet{Ghosal}, the contraction rate $\epsilon_n=n^{-\beta \gamma /(2\beta \gamma+\beta +r\gamma)}$ follows. 
\end{proof}

\begin{proof}[Proof of Theorem~\ref{posterior adaptation}] In view of Section~8.3 of \citet{Ghosal} and the proof of Theorem~\ref{posterior known smoothness}, it suffices to show that, for some `pre-rate' $\bar \epsilon_n\le \epsilon_n$, the prior concentration condition $-\log \Pi (\|f-f^*\|_n \le \bar\epsilon_n)\lesssim n\bar\epsilon_n^2$ and construct a sieve $\mathcal{S}_n\subset \HH^n$ such that $\log N(\epsilon_n, \mathcal{S}_n, \|\cdot\|_n)\lesssim n\epsilon_n^2$ and $\Pi (\mathcal{S}_n^c)\lesssim e^{-A n\bar\epsilon_n^2}$ for some sufficiently large constant $A>0$. 
	
	We represent $f^*$ in terms of the orthonormal basis as $f^* =\sum_{i=1}^{n} \sum_{j=1}^\infty \kappa_j^{1/2} \vartheta_{ij} \psi_j e_j$. Since  $f^*\in  \mathcal{H}^{\beta,\gamma}(Q)$, we have that $n^{-1} \sum_{i=1}^{n}  j^{2\gamma}(\lambda_i +n^{-2})^2 \vartheta_{ij}^2\le Q^2$. 
	Define $f^*_{I_n,J_n}=\sum_{i=1}^{I_n} \sum_{j=1}^{J_n} \kappa_j^{1/2} \vartheta_{ij} \psi_j e_j$, where $I_n,J_n\to \infty$ are chosen so that 
	$$\|f^*-f^*_{I_n,J_n}\|_n^2 \le \sum_{(i,j): i> I_n\; \mathrm{or}\; j>J_n} \vartheta_{ij}^2 \le (I_n^{-2\beta/r}+J_n^{-2\gamma}) \sum_{i=1}^{n} \sum_{j=1}^\infty (\lambda_i+n^{-2})^\beta j^{2\gamma} \vartheta_{ij}^2\lesssim \bar\epsilon_n^2,$$
	which can be ensured by the choices $I_n\asymp \epsilon_n^{-r/\beta}$ and $J_n\asymp n^{-1/\gamma}$. As in the proof of Lemma~\ref{RKHS}, the estimation of the prior probability concentration reduces to that of the $\bar\epsilon_n$-ball probability of independent normal distribution in the $I_n J_n$-dimensional Euclidean space, and gives $-\log \Pi (\|f-f^*\|_n \le \bar\epsilon_n)\lesssim -\log \pi(I_n,J_n) + I_n J_n\log (1/\bar\epsilon_n)\lesssim I_n J_n \log n \lesssim \bar \epsilon_n^{-(r/\beta)-1/\gamma}\log n$.  
	
	Let $A'>0$, to be chosen later. Define the sieve $\mathcal{S}_n= \{f= \sum_{i=1}^{I} \sum_{j=1}^{J} \kappa_j^{-1/2} \eta_{ij} \psi_j e_j,\; IJ\le A'I_n J_n, \, \max |\eta_{ij}|\le \sqrt{n} \}$. 
	Then $\Pi (\mathcal{S}_n^c)$ is bounded by a constant times $$ \pi(IJ> A' I_n J_n)+ A' I_n J_n \P (|Z|>\sqrt{n})  \lesssim (I_n J_n)^2 e^{-a_2 A' I_n J_n (\log I_n+\log J_n)}+e^{\log I_n+\log J_n-n/2},$$ 	
	where $Z$ stands for a standard normal variable; here we have used the fact that the number of ways a product $IJ$ is obtained is at most $(IJ)^2$. By choosing $A'$ large enough, the probability can be bounded by $\exp\{-A \bar \epsilon_n^{-(r/\beta)-1/\gamma}\}$ for any given $A>0$, so the required condition on $\Pi(\mathcal{S}_n^c)$ is satisfied by setting $\epsilon_n^{-(r/\beta)-1/\gamma}=n \bar \epsilon_n^2$, that is, $\bar \epsilon_n=(n/\log n)^{-\beta \gamma /(2\beta \gamma +\beta +r\gamma)}$. 
	
	Finally, it remains to estimate the entropy of $\mathcal{S}_n$, which is a union of the centered $IJ$-dimensional cube of diameter $2\sqrt{n}$ for $IJ\le  A'I_n J_n$. By standard arguments, the covering number is estimated by $\sum_{(I,J): IJ\le A' I_n J_n} (2 \sqrt{n}/\epsilon_n)^{IJ}\lesssim  I_n^2 J_n^2 (2 \sqrt{n} /\epsilon_n)^{A' I_n J_n}$, and hence  
	$$\log N(\epsilon_n, \mathcal{S}_n, \|\cdot\|_n)\lesssim \log I_n+\log J_n+ I_n J_n \log (2 \sqrt{n} /\epsilon_n) \lesssim  \bar \epsilon_n^{-(r/\beta)-1/\gamma}\log n = n \bar \epsilon_n^2 \log n \le n\epsilon_n^2$$ for $\epsilon_n= (\log n)^{1/2}\bar \epsilon_n=n^{-\beta \gamma /(2\beta \gamma +\beta +r\gamma)}(\log n)^{1/2+\beta \gamma /(2\beta \gamma +\beta +r\gamma) }$, establishing the rate. 
\end{proof}

\begin{proof}[Proof of Theorem~\ref{coverage}]
	We work with the canonical form with observations $Z_{ij}\sim \mathrm{N} ( \vartheta_{ij},1)$ with parameter $\vartheta=(\!( \vartheta_{ij})\!)$. Observe that the squared norm $\|f\|_n^2$ reduces to $n^{-1}\sum_{i=1}^{n}\sum_{j=1}^\infty \vartheta_{ij}^2$ and the prior can be expressed as $\vartheta_{ij}\sim \N (0,d_{ij}^2)$ independently, where $d_{ij}^2= c^{2\beta/r}i^{-2\beta/r} j^{-(2\gamma+1)}$. Hence the posterior distribution is given by $\vartheta_{ij}|Z_{ij} \sim \N ( Z_{ij}/(1+d_{ij}^{-2}), 1/(1+d_{ij}^{-2})) $, and the Bayes estimator for $\vartheta_{ij}$ is $\hat \vartheta_{ij}= Z_{ij}/(1+d_{ij}^{-2})$.  Then a natural $(1-\tau)$-posterior credible ball for $\vartheta$ around $\hat \vartheta=(\!(\hat \vartheta_{ij})\!)$ is given by $\{\vartheta: \|\vartheta- \hat \vartheta\|^2\le q_\tau\}$, where $q_\tau$, $0<\tau<1$, is the posterior $(1-\tau)$-quantile of $\|\vartheta- \hat \vartheta\|^2$. It is to be observed that $q_\tau$ is deterministic. The assertion of the theorem then reduces to showing that, for $\vartheta^*$ standing for the true value of $\vartheta$, 
	\begin{itemize} 
		\item [(i)] 
		$q_\tau \asymp n^{1-2\beta\gamma/(2\beta \gamma+\beta+r \gamma)}$, 
		\item [(ii)] 
		for any $Q>0$, there exists $K>0$ such that $\P_{\vartheta^*} (\|\vartheta^*-\hat\vartheta\|^2>K n^{1-2\beta\gamma/(2\beta \gamma+\beta+r \gamma)} )\to 0$. 
	\end{itemize}
	To establish (i), let $U=\|\vartheta-\hat \vartheta\|^2$, whose posterior distribution is deterministic. By Chebyshev's inequality, it suffices to show that $\E (U) \asymp n^{1-2\beta\gamma/(2\beta \gamma+\beta+r \gamma)}$ and 
	$\mathrm{var} (U) \ll n^{2-4\beta\gamma/(2\beta \gamma+\beta+r \gamma)}$, because then $q_\tau\asymp \E(U)$. 
	As $U$ is a sum of squares of independent mean-zero normal variables, a relation between the fourth central moment and the variance of a normal distribution, reduces the assertions respectively to  
	\begin{itemize} 
		\item [(iii)] 
		$ \sum_{i=1}^{n} \sum_{j=1}^\infty (1+d_{ij}^{-2})^{-1}\asymp n^{1-2\beta\gamma/(2\beta \gamma+\beta+r \gamma)}$, 
		\item [(iv)] 
		$\sum_{i=1}^{n} \sum_{j=1}^\infty (1+d_{ij}^{-2})^{-2} \ll n^{2-4\beta\gamma/(2\beta \gamma+\beta+r \gamma)}$. 
	\end{itemize} 
	To see when $1$ or $d_{ij}^{-2}$ dominates in the sum, introduce $\mathcal{K}=\{ (i,j): d_{ij}^{-2}\le 1\}$ and note that 
	$$\#\mathcal{K}=\#\{(i,j): i\le c j^{-r(2\gamma+1)/(2\beta)}, j\le c^{2\beta/(r(2\gamma+1))} \} = c \sum_{j=1}^{c^{2\beta/(r(2\gamma+1))}} j^{-r(2\gamma+1)/(2\beta)}\asymp c$$ 
	because by the assumption $\gamma > (\beta/r)-1/2>0$, the power of $j$ is less than $-1$, making the series summable. To obtain an upper bound for the sum in (iii), we use the dominant term $1$ for $(i,j)\in \mathcal{K}$ and $d_{ij}^{-2}$ for $(i,j)\in \mathcal{K}^c$, giving   $ \sum_{i=1}^{n} \sum_{j=1}^\infty (1+d_{ij}^{-2})^{-1}\le  \# \mathcal{K}+ c^{2\beta/r} \sum_{(i,j)\in \mathcal{K}^c} i^{-2\beta/r} j^{-1-2\gamma}$. The first term is bounded by a multiple of $c$, while the second term can be split in the sum of  
	$c^{2\beta/r}\sum_{i\ge 1}  i^{-2\beta/r} \sum_{j>c^{2\beta/(r(2\gamma+1))}} j^{-(2\gamma+1)}$ and $c^{2\beta/r}\sum_{j \le c^{2\beta /(r(2\gamma+1))}} j^{-(2\gamma+1)}  \sum_{i> c j^{-(2\gamma+1)r/(2\beta)}}    i^{-2\beta/(r(2\gamma+1))}$.  
	Because $2\beta < r (2\gamma+1)$ by the assumption  $\gamma> (\beta/r)-1/2>0$, the former is estimated as $c^{2\beta/r} (c^{2\beta/(r(2\gamma+1))})^{-2\gamma}=c^{2\beta /(r(2\gamma+1))}\le c=n^{(\beta+r \gamma)/(2\beta\gamma+\beta+r\gamma)}$. The latter is estimated to be 
	$c^{2\beta/r}\sum_{j< c^{2\beta/(r(2\gamma+1))}} j^{-(2\gamma+1)} \sum_{i > c j^{-(2\gamma+1)r/(2\beta)}}  i^{-2\beta/r}$ which can be written as  $$c^{2\beta/r}\sum_{j< c^{2\beta/(r(2\gamma+1))}} j^{-(2\gamma+1)}  (c j^{-(2\gamma+1)r/(2\beta)})^{1-2\beta/r}=c \sum_{j< c^{2\beta/(r(2\gamma+1))}} j^{-(2\gamma+1)r/(2\beta)}.$$  
	This simplifies to $c (c^{2\beta/(r(2\gamma+1))})^{1-(2\gamma+1)r/(2\beta)}=c^{2\beta/(r(2\gamma+1))}\le c=n^{(\beta+r \gamma)/(2\beta\gamma+\beta+r\gamma)},$  
	again by the assumption $\gamma> (\beta/r)-1/2>0$. This proves the statement in (iii). It may be noted that the lower bound is within a factor one-half of the upper bound, because to lower bound, we use the sum in the denominator instead of the dominant term, which is at least half of the sum. Thus, (iii) holds.

	The estimate (iv), we follow previous decomposition again $\sum_{i=1}^{n} \sum_{j=1}^\infty (1+d_{ij}^{-2})^{-2}\le  \# \mathcal{K}+ c^{4\beta/r} \sum_{(i,j)\in \mathcal{K}^c} i^{-4\beta/r} j^{-2-4\gamma}$. We split the second part as $c^{4\beta/r}\sum_{i\ge 1}  i^{-4\beta/r} \sum_{j>c^{2\beta/(r(2\gamma+1))}} j^{-(4\gamma+2)}$ and $c^{4\beta/r}\sum_{j \le c^{2\beta /(r(2\gamma+1))}} j^{-(4\gamma+2)}  \sum_{i> c j^{-(2\gamma+1)r/(2\beta)}}    i^{-4\beta/(r(2\gamma+1))}$. Using the fact that $2\beta<r(2\gamma+1)$, we can show that each of the two terms are upper bounded by $c^{\frac{2\beta}{r(2\gamma+1)}}$. Thus, $\sum_{i=1}^{n} \sum_{j=1}^\infty (1+d_{ij}^{-2})^{-2}\lesssim c\ll n^{2-4\beta\gamma/(2\beta \gamma+\beta+r \gamma)}$.

	We now address (ii). By Chebyshev's inequality and a bias-variance decomposition, it suffices to show that $\|\E_{\vartheta^*}(\hat\vartheta)-\vartheta^*\|^2\lesssim  n^{1-2\beta\gamma/(2\beta \gamma+\beta+r \gamma)}$ and $\sum_{i=1}^{n} \sum_{j=1}^\infty \mathrm{var}_{\vartheta^*} (\hat \vartheta_{ij})\lesssim  n^{1-2\beta\gamma/(2\beta \gamma+\beta+r \gamma)}$ uniformly on $\mathcal{H}^{\beta,\gamma}(Q)$. The latter follows from (iii) by observing that $\mathrm{var}_{\vartheta^*} (\hat \vartheta_{ij})=(1+d_{ij}^{-2})^{-2}\le (1+d_{ij}^{-2})^{-1}$. The squared bias term is given by 
	\begin{align}
	\label{neweqn3}
	\sum_{i=1}^{n} \sum_{j=1}^\infty (d_{ij}^{2}+1)^{-2} (\vartheta_{ij}^*)^2\le   \sum_{ (i,j)\in \mathcal{K}} c^{-4\beta/r} i^{4\beta/r} j^{4\gamma+2}(\vartheta_{ij}^*)^2 +  \sum_{(i,j)\in \mathcal{K}^c}(\vartheta_{ij}^*)^2.
	\end{align}
	The first term is bounded by 
	$c^{-4\beta/r} \max\{ i^{2\beta/r} j^{2\gamma+2}: (i,j)\in \mathcal{K}\} \sum_{i=1}^{n} \sum_{j=1}^\infty i^{2\beta/r} j^{2\gamma} (\vartheta_{ij}^*)^2$, 
	which is further bounded by $$ Q^2 n c^{-4\beta/r} c^{2\beta/r} c^{2\beta/(r(2\gamma+1))}=Q^2 nc^{-(2\beta/r)(2\gamma/(2\gamma+1)) }=Q^2 n^{1-[2\beta\gamma/(2\beta\gamma+\beta+r\gamma)]2(\beta+r\gamma)/(r(2\gamma+1))}$$ 
	since $\gamma> (\beta/r)- 1/2>0$. Thus it follows that this term is bounded by a multiple of 
	$n^{1-2\beta\gamma/(2\beta\gamma+\beta+r\gamma)}$, as asserted. 
	
	The second term is bounded by $\max \{i^{-2\beta/r}j^{-2\gamma}: (i,j)\in \mathcal{K}^c \}  \sum_{i=1}^{n}\sum_{j=1}^\infty i^{2\beta}j^{2\gamma} (\vartheta_{ij}^*)^2$ 
	which is dominated by the expression $Q^2 n (\max \{i^{-2\beta/r}j^{-(2\gamma+1)}: (i,j)\in \mathcal{K}^c \})^{2\gamma/(2\gamma+1)}\lesssim n (c^{-2\beta/r})^{2\gamma/(2\gamma+1)}.$   
	The last expression is  
	$n^{1-4\beta\gamma (\beta+r\gamma)/(r(2\beta\gamma+\beta+r\gamma))}\le n^{1-2\beta\gamma /(2\beta\gamma+\beta+r\gamma)}$, as asserted; here we have used the characterization of $\mathcal{K}^c$ as the collection of pairs $(i,j)$ such that $i^{-2\beta/r} j^{-(2\gamma+1)}\le c^{-2\beta/r}$ and that $2(\beta+r\gamma)\ge r (2\gamma+1)$ by the assumption that $\gamma> (\beta/r)- 1/2>0$.	
\end{proof}	

\begin{proof}[Proof of Theorem~\ref{discrete}]
	To simplify the expressions, we assume that the known value of the error standard deviation $\sigma$ is 1. We can write the data in the matrix form $\mathrm{Y}=(\!(Y_i(k/T) : i=1,\ldots,n; k=1,\ldots,T)\!)$. 
	Observe that $\mathrm{W}:=(\!(\psi_{j}(k/T): j,k=1,\ldots, T)\!)$ is an orthogonal matrix.
	We can re-write the model in terms of $\mathrm{Z}:= \mathrm{B} \mathrm{Y} \mathrm{W}$, where $\mathrm{B}$ is the orthogonal matrix formed by the normalized eigenvectors of $\mathrm{L}$.  Then the $(i,j)$th entry of $\mathrm{Z}$ satisfies $Z_{ij} \sim \N (\vartheta_{ij},1)$ independently. Thus the posterior distribution is given by  $ \vartheta_{ij}|Z_{ij}\sim  \N( Z_{ij}/(1+d_{ij}^{-2}), 1/(1+d_{ij}^{-2}))$, where $d_{ij}=c^{2\beta/r}i^{-2\beta/r} j^{-(2\gamma+1)}$. By a standard bias-variance decomposition as in the proof of Theorem~\ref{coverage}, it follows that 
	\begin{align*} 
	\E_{\vartheta^*} \E (d_n^2(\vartheta,\vartheta^*)|\mathrm{Z})=\frac{1}{nT}\big \{  \sum_{i=1}^n \sum_{j=1}^T (1+d_{ij}^{-2})^{-2}+ \sum_{i=1}^n \sum_{j=1}^T \frac{d_{ij}^{-4} (\vartheta_{ij}^*)^2}{(1+d_{ij}^{-2})^2}+\sum_{i=1}^n \sum_{j=1}^T  (1+d_{ij}^{-2})^{-1}\}.
	\end{align*}
	The first term is bounded by the third term. 
	We follow the arguments used in the proof of Theorem~\ref{coverage} to bound each sum above by breaking the index set in $\mathcal{K}=\{(i,j):d_{ij}^{-2}\leq 1\}$ and its complement.
	By arguments there, $\#\mathcal{K}\lesssim c$, $c\log c$ or $c^{2\beta/(\gamma(2\gamma+1))}$ respectively for $2\beta$ less than, equal to or greater than $r(2\gamma+1)$. Then the third  term is estimated as a constant multiple of $(nT)^{-1}  (c+c^{{2\beta}/({r(2\gamma+1)})})$ if $2\beta\ne r(2\gamma+1)$ and $(nT)^{-1}c\log c$ if $2\beta= r(2\gamma+1)$. 
	Finally the second term is bounded by the sum of  $(nT)^{-1} \sum_{(i,j)\in \mathcal{K}} d_{ij}^{-4} (\vartheta_{ij}^*)^2$ and $(nT)^{-1} \sum_{(i,j)\in \mathcal{K}^c} (\vartheta_{ij}^*)^2$. Putting the value of $d_{ij}$, the former is bounded by a constant multiple of 
	$$ c^{-4\beta/r} \max_{(i,j)\in \mathcal{K}} { i^{2\beta/r} j^{2\gamma+2} } \le c^{-4\beta/r} c^{2\beta/r}  \max_{(i,j)\in \mathcal{K}}  j \le c^{-2\beta/r}  (c^{2\beta/r})^{1/(2\gamma+1)}=c^{-4\beta\gamma/(r(2\gamma+1)) }.$$ 
	The latter is bounded by a constant multiple of $\max \{ i^{-2\beta/r} j^{-2\gamma}: (i,j)\in \mathcal{K}^c\}\le ( c^{2\beta/r})^{-2\gamma/(2\gamma+1)}=c^{-4\beta\gamma/(r(2\gamma+1)) }$. Putting all these together, for $2\beta\ne r(2\gamma+1)$, the square of the rate is given by 
	\begin{align} 
	\label{neweqn4}
	(nT)^{-1}c+(nT)^{-1} c^{{2\beta}/({r(2\gamma+1)})}+c^{-4\beta\gamma/(r(2\gamma+1)) }.
	\end{align}	
	The expression is modified to $(nT)^{-1}c\log c+c^{-2\gamma}$ for $2\beta= r(2\gamma+1)$.
	If $2\beta< r(2\gamma+1)$, the second term in \eqref{neweqn4} is dominated by the first. The latter matches the third term for the choice $c=(nT)^{r(2\gamma+1)/(4\beta\gamma+2r\gamma+r)}$, yielding the stated rate $\epsilon_n=(nT)^{-2\beta\gamma/(4\beta\gamma+2r\gamma+r)}$. When 
	$2\beta> r(2\gamma+1)$, the second term in \eqref{neweqn4} is larger and matches the third
	giving the rate $\epsilon_n=(nT)^{-\gamma/(2\gamma+1)}$ for the choice $c=(nT)^{r/2\beta}$. For $2\beta= r(2\gamma+1)$, the rate $(nT/\log (nT))^{-2\gamma/(2\gamma+1)}$ is obtained upon choosing $c=(nT/\log (nT))^{1/(2\gamma+1)}$. 
	
	The last part of the theorem follows by observing that if all true functions are uniformly Lipschitz continuous, the supremum distance between a function and its reconstruct from its values at the grid-points $k/T$, $k=1,\ldots,T$, through linear interpolation, is uniformly of the order $T^{-1}$. 
\end{proof}

\bibliographystyle{plainnat}
\bibliography{Laplacian_reference}
\end{document}